\documentclass[onecolumn, 11pt]{IEEEtran}
\usepackage{amsmath}
\usepackage{amssymb}
\usepackage{amsfonts}
\usepackage{graphicx}
\usepackage{epsfig}
\usepackage{subfigure}
\usepackage{psfrag}

\title{MIMO Broadcasting for Simultaneous Wireless Information and Power Transfer\footnote{This paper has been presented in part at IEEE
Global Communications Conference (Globecom), December 5-9, 2011,
Houston, USA.} \footnote{R. Zhang is with the Department of
Electrical and Computer Engineering, National University of
Singapore (e-mail:elezhang@nus.edu.sg). He is also with the
Institute for Infocomm Research, A*STAR, Singapore.} \footnote{C. K.
Ho is with the Institute for Infocomm Research, A*STAR, Singapore
(e-mail:hock@i2r.a-star.edu.sg).}}

\author{Rui Zhang and Chin Keong Ho}

\begin{document}
\maketitle \thispagestyle{empty}

\begin{abstract}
Wireless power transfer (WPT) is a promising new solution to provide
convenient and perpetual energy supplies to wireless networks. In
practice, WPT is implementable by various technologies such as
inductive coupling, magnetic resonate coupling, and electromagnetic
(EM) radiation, for short-/mid-/long-range applications,
respectively. In this paper, we consider the EM or radio signal
enabled WPT in particular. Since radio signals can carry energy as
well as information at the same time, a unified study on
\emph{simultaneous wireless information and power transfer} (SWIPT)
is pursued. Specifically, this paper studies a multiple-input
multiple-output (MIMO) wireless broadcast system consisting of three
nodes, where one receiver harvests energy and another receiver
decodes information separately from the signals sent by a common
transmitter, and all the transmitter and receivers may be equipped
with multiple antennas. Two scenarios are examined, in which the
information receiver and energy receiver are \emph{separated} and
see different MIMO channels from the transmitter, or
\emph{co-located} and see the identical MIMO channel from the
transmitter. For the case of separated receivers, we derive the
optimal transmission strategy to achieve different tradeoffs for
maximal information rate versus energy transfer, which are
characterized by the boundary of a so-called \emph{rate-energy}
(R-E) region. For the case of co-located receivers, we show an outer
bound for the achievable R-E region due to the potential limitation
that practical energy harvesting receivers are not yet able to
decode information directly. Under this constraint, we investigate
two practical designs for the co-located receiver case, namely
\emph{time switching} and \emph{power splitting}, and characterize
their achievable R-E regions in comparison to the outer bound.
\end{abstract}


\setlength{\baselineskip}{1.3\baselineskip}
\newtheorem{definition}{\underline{Definition}}[section]
\newtheorem{fact}{Fact}
\newtheorem{assumption}{Assumption}
\newtheorem{theorem}{\underline{Theorem}}[section]
\newtheorem{lemma}{\underline{Lemma}}[section]
\newtheorem{corollary}{Corollary}[section]
\newtheorem{proposition}{\underline{Proposition}}[section]
\newtheorem{example}{\underline{Example}}[section]
\newtheorem{remark}{\underline{Remark}}[section]
\newtheorem{algorithm}{\underline{Algorithm}}[section]
\newcommand{\mv}[1]{\mbox{\boldmath{$ #1 $}}}
\newcommand{\smv}[1]{\small{\mbox{\boldmath{$ #1 $}}}}

\section{Introduction}

Energy-constrained wireless networks, such as sensor networks, are
typically powered by batteries that have limited operation time.
Although replacing or recharging the batteries can prolong the
lifetime of the network to a certain extent, it usually incurs high
costs and is inconvenient, hazardous (say, in toxic environments),
or even impossible (e.g., for sensors embedded in building
structures or inside human bodies). A more convenient, safer, as
well as ``greener'' alternative is thus to harvest energy from the
environment, which virtually provides perpetual energy supplies to
wireless devices. In addition to other commonly used energy sources
such as solar and wind, ambient radio-frequency (RF) signals can be
a viable new source for energy scavenging. It is worth noting that
RF-based energy harvesting is typically suitable for low-power
applications (e.g., sensor networks), but also can be applied for
scenarios with more substantial power consumptions if dedicated
wireless power transmission is implemented.\footnote{Interested
readers may visit the company website of Powercast at
http://www.powercastco.com/ for more information on recent
applications of dedicated RF-based power transfer.}

On the other hand, since RF signals that carry energy can at the
same time be used as a vehicle for transporting information,
\emph{simultaneous wireless information and power transfer} (SWIPT)
becomes an interesting new area of research that attracts increasing
attention. Although a unified study on this topic is still in the
infancy stage, there have been notable results reported in the
literature \cite{Varshney,Grover}. In \cite{Varshney}, Varshney
first proposed a \emph{capacity-energy} function to characterize the
fundamental tradeoffs in simultaneous information and energy
transfer. For the single-antenna or SISO (single-input
single-output) AWGN (additive white Gaussian noise) channel with
amplitude-constrained inputs, it was shown in \cite{Varshney} that
there exist nontrivial tradeoffs in maximizing information rate
versus (vs.) power transfer by optimizing the input distribution.
However, if the average transmit-power constraint is considered
instead, the above two goals can be shown to be aligned for the SISO
AWGN channel with Gaussian input signals, and thus there is no
nontrivial tradeoff. In \cite{Grover}, Grover and Sahai extended
\cite{Varshney} to frequency-selective single-antenna AWGN channels
with the average power constraint, by showing that a non-trivial
tradeoff exists in frequency-domain power allocation for maximal
information vs. energy transfer.

As a matter of fact, \emph{wireless power transfer} (WPT) or in
short \emph{wireless power}, which generally refers to the
transmissions of electrical energy from a power source to one or
more electrical loads without any interconnecting wires,  has been
investigated and implemented with a long history. Generally
speaking, WPT is carried out using either the ``near-field''
electromagnetic (EM) induction (e.g., inductive coupling, capacitive
coupling) for short-distance (say, less than a meter) applications
such as passive radio-frequency identification (RFID) \cite{Want},
or the ``far-field'' EM radiation in the form of microwaves or
lasers for long-range (up to a few kilometers) applications such as
the transmissions of energy from orbiting solar power satellites to
Earth or spacecrafts \cite{NASA}. However, prior research on EM
radiation based WPT, in particular over the RF band, has been
pursued independently from that on wireless information transfer
(WIT) or radio communication. This is non-surprising since these two
lines of work in general have very different research goals: WIT is
to maximize the {\it information transmission capacity} of wireless
channels subject to channel impairments such as the fading and
receiver noise, while WPT is to maximize the {\it energy
transmission efficiency} (defined as the ratio of the energy
harvested and stored at the receiver to that consumed by the
transmitter) over a wireless medium. Nevertheless, it is worth
noting that the design objectives for WPT and WIT systems can be
aligned, since given a transmitter energy budget, maximizing the
signal power received (for WPT) is also beneficial in maximizing the
channel capacity (for WIT) against the receiver noise.

Hence, in this paper we attempt to pursue a unified study on WIT and
WPT for emerging wireless applications with such a dual usage. An
example of such wireless dual networks is envisaged in Fig.
\ref{fig:system model new}, where a fixed access point (AP)
coordinates the two-way communications to/from a set of distributed
user terminals (UTs). However, unlike the conventional wireless
network in which both the AP and UTs draw energy from constant power
supplies (by e.g. connecting to the grid or a battery), in our
model, only the AP is assumed to have a constant power source, while
all UTs need to replenish energy from the received signals sent by
the AP via the far-field RF-based WPT. Consequently, the AP needs to
coordinate the wireless information and energy transfer to UTs in
the downlink, in addition to the information transfer from UTs in
the uplink. Wireless networks with such a dual information and power
transfer feature have not yet been studied in the literature to our
best knowledge, although some of their interesting applications have
already appeared in, e.g., the body sensor networks \cite{FeiZhang}
with the out-body local processing units (LPUs) powered by battery
communicating and at the same time sending wireless power to in-body
sensors that have no embedded power supplies. However, how to
characterize the fundamental information-energy transmission
tradeoff in such dual networks is still an open problem.

\begin{figure}
\begin{center}
\scalebox{0.6}{\includegraphics*[81pt,438pt][484pt,672pt]{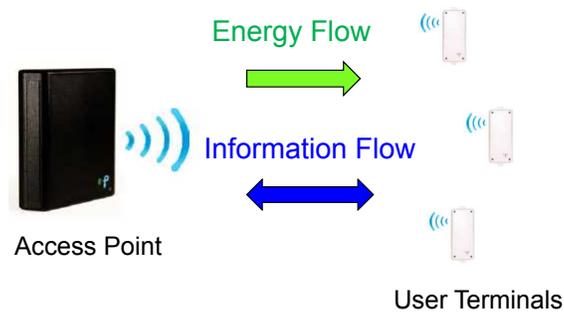}}
\end{center}
\caption{A wireless network with dual information and energy
transfer.}\label{fig:system model new}
\end{figure}

In this paper, we focus our study on the downlink case with
simultaneous WIT and WPT from the AP to UTs. In the generic system
model depicted in Fig. \ref{fig:system model new}, each UT can in
general harvest energy and decode information at the same time (by
e.g. applying the power splitting scheme introduced later in this
paper). However, from an implementation viewpoint, one particular
design whereby each UT operates as either an information receiver or
an energy receiver at any given time may be desirable, which is
referred to as {\it time switching}. This scheme is practically
appealing since state-of-the-art wireless information and energy
receivers are typically designed to operate separately with very
different power sensitivities (e.g., $-50$dBm for information
receivers vs. $-10$dBm for energy receivers). As a result, if time
switching is employed at each UT jointly with the ``near-far'' based
transmission scheduling at the AP, i.e., UTs that are close to the
AP and thus receive high power from the AP are scheduled for WET,
whereas those that are more distant from the AP and thus receive
lower power are scheduled for WIT, then SWIPT systems can be
efficiently implemented with existing information and energy
receivers and the additional time-switching device at each receiver.

For an initial study on SWIPT, this paper considers the simplified
scenarios with only one or two active UTs in the network at any
given time. For the case of two UTs, we assume time switching, i.e.,
the two UTs take turns to receive energy or (independent)
information from the AP over different time blocks. As a result,
when one UT receives information from the AP, the other UT can
opportunistically harvest energy from the same signal broadcast by
the AP, and vice versa. Hence, at each block, one UT operates as an
information decoding (ID) receiver, and the other UT as an energy
harvesting (EH) receiver. We thus refer to this case as {\it
separated EH and ID receivers}. On the other hand, for the case with
only one single UT to be active at one time (while all other UTs are
assumed to be in the off/sleep mode), the active UT needs to harvest
energy as well as decode information from the same signal sent by
the AP, i.e., the same set of receiving antennas are shared by both
EH and ID receivers residing in the same UT. Thus, this case is
referred to as {\it co-located EH and ID receivers}. Surprisingly,
as we will show later in this paper, the optimal information-energy
tradeoff for the case of co-located receivers is more challenging to
characterize than that for the case of separated receivers, due to a
potential limitation that practical EH receiver circuits are not yet
able to decode the information directly and vice versa. Note that
similar to the case of separated receives, time switching can also
be applied in the case of co-located receivers to orthogonalize the
information and energy transmissions at each receiving antenna;
however, this scheme is in general suboptimal for the achievable
rate-energy tradeoffs in the case of co-located receivers, as will
be shown later in this paper.

Some further assumptions are made in this paper for the purpose of
exposition. Firstly, this paper considers a quasi-static fading
environment where the wireless channel between the AP and each UT is
assumed to be constant over a sufficiently long period of time
during which the number of transmitted symbols can be approximately
regarded as being infinitely large. Under this assumption, we
further assume that it is feasible for each UT to estimate the
downlink channel from the AP and then send it back to the AP via the
uplink, since the time overhead for such channel estimation and
feedback is a negligible portion of the total transmission time due
to quasi-static fading. We will address the more general case of
fading channels with imperfect/partial channel knowledge at the
transmitter in our future work. Secondly, we assume that the system
under our study typically operates at the high signal-to-noise ratio
(SNR) regime for the ID receiver in the case of co-located
receivers. This is to be compatible with the high-power operating
requirement for the EH receiver of practical interest as previously
mentioned. Thirdly, without loss of generality, we assume a
multi-antenna or MIMO (multiple-input multiple-output) system, in
which the AP is equipped with multiple antennas, and each UT is
equipped with one or more antennas, for enabling both the
high-performance wireless energy and information transmissions (as
it is well known that for WIT only, MIMO systems can achieve folded
array/capacity gains over SISO systems by spatial
beamforming/multiplexing \cite{Telater}).

\begin{figure}
\begin{center}
\scalebox{0.6}{\includegraphics*[48pt,410pt][418pt,695pt]{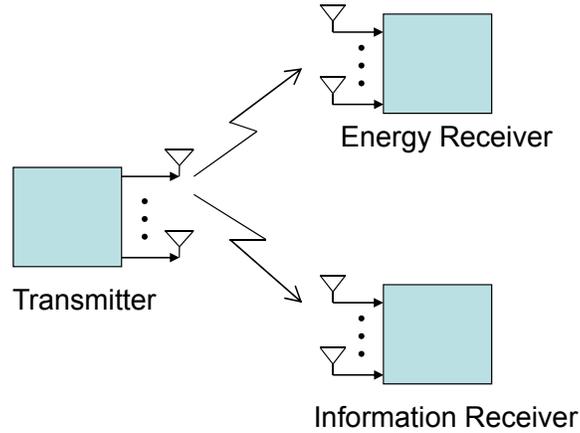}}
\end{center}
\caption{A MIMO broadcast system for simultaneous wireless
information and power transfer.}\label{fig:system model}
\end{figure}

Under the above assumptions, a three-node MIMO broadcast system is
considered in this paper, as shown in Fig. \ref{fig:system model},
wherein the EH and ID receivers harvest energy and decode
information separately from the signal sent by a common transmitter.
Note that this system model refers to the case of separated EH and
ID receivers in general, but includes the co-located receivers as a
special case when the MIMO channels from the transmitter to both
receivers become identical. Assuming this model, the main results of
this paper are summarized as follows:
\begin{itemize}
\item For the case of separated EH and ID receivers, we
design the optimal transmission strategy to achieve different
tradeoffs between maximal information rate vs. energy transfer,
which are characterized by the boundary of a so-called
\emph{rate-energy} (R-E) region. We derive a semi-closed-form
expression for the optimal transmit covariance matrix (for the joint
precoding and power allocation) to achieve different rate-energy
pairs on the boundary of the R-E region. Note that the R-E region is
a multiuser extension of the single-user capacity-energy function in
\cite{Varshney}. Also note that the multi-antenna broadcast channel
(BC) has been investigated in e.g. \cite{Carie}--\cite{Luo2} for
information transfer solely by unicasting or multicasting. However,
MIMO-BC for SWIPT as considered in this paper is new and has not yet
been studied by any prior work.

\item For the case of co-located EH and ID receivers, we show that the proposed
solution for the case of separated receivers is also applicable with
the identical MIMO channel from the transmitter to both ID and EH
receivers. Furthermore, we consider a potential practical constraint
that EH receiver circuits cannot directly decode the information
(i.e., any information embedded in received signals sent to the EH
receiver is lost during the EH process). Under this constraint, we
show that the R-E region with the optimal transmit covariance
(obtained without such a constraint) in general only serves as a
performance outer bound for the co-located receiver case.

\item Hence, we investigate two practical receiver designs,
namely \emph{time switching} and \emph{power splitting}, for the
case of co-located receivers. As shown in Fig. \ref{fig:practical
schemes}, for time switching, each receiving antenna periodically
switches between the EH receiver and ID receiver, whereas for power
splitting, the received signal at each antenna is split into two
separate signal streams with different power levels, one sent to the
EH receiver and the other to the ID receiver.  Note that time
switching has also been proposed in \cite{Varshney Thesis} for the
SISO AWGN channel. Furthermore, note that the \emph{antenna
switching} scheme whereby the receiving antennas are divided into
two groups with one group switched to information decoding and the
other group to energy harvesting can be regarded as a special case
of power splitting with only binary splitting power ratios at each
receiving antenna. For these practical receiver designs, we derive
their achievable R-E regions as compared to the R-E region outer
bound, and characterize the conditions under which their performance
gaps can be closed. For example, we show that the power splitting
scheme approaches the tradeoff upper bound asymptotically when the
RF-band antenna noise at the receiver becomes more dominant over the
baseband processing noise (more details are given in Section IV-C).
\end{itemize}

\begin{figure}
\begin{center}
\scalebox{0.6}{\includegraphics*[70pt,545pt][530pt,725pt]{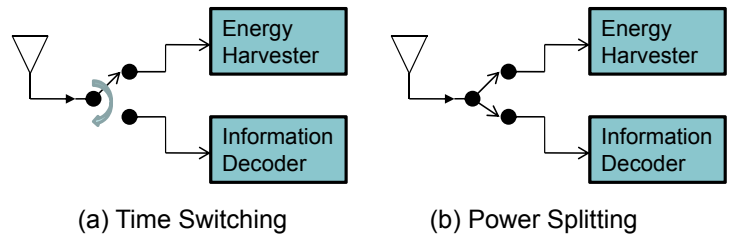}}
\end{center}
\caption{Two practical designs for the co-located energy and
information receivers, which are applied for each receiving antenna.
}\label{fig:practical schemes}
\end{figure}

The rest of this paper is organized as follows: Section
\ref{sec:system model} presents the system model, characterizes the
rate-energy region, and formulates the problem for finding the
optimal transmit covariance matrix. Section \ref{sec:separated
receivers} presents the optimal transmit covariance solution for the
case of separated receivers. Section \ref{sec:colocated receivers}
extends the solution to the case of co-located receivers to obtain a
performance upper bound, proposes practical receiver designs, and
analyzes their performance limits as compared to the performance
upper bound. Finally, Section \ref{sec:conclusion} concludes the
paper and provides some promising directions for future work.

{\it Notation}: For a square matrix $\mv{S}$, $\mathtt{tr}(\mv{S})$,
$|\mv{S}|$, $\mv{S}^{-1}$, and $\mv{S}^{\frac{1}{2}}$ denote its
trace, determinant, inverse, and square-root, respectively, while
$\mv{S}\succeq 0$ and $\mv{S}\succ 0$ mean that $\mv{S}$ is positive
semi-definite and positive definite, respectively. For an
arbitrary-size matrix $\mv{M}$, $\mv{M}^{H}$ and $\mv{M}^{T}$ denote
the conjugate transpose and transpose of $\mv{M}$, respectively.
$\mathtt{diag}(x_1, \ldots, x_M)$ denotes an $M \times M$ diagonal
matrix with $x_1,\ldots,x_M$ being the diagonal elements. $\mv{I}$
and $\mv{0}$ denote an identity matrix and an all-zero vector,
respectively, with appropriate dimensions. $\mathbb{E}[\cdot]$
denotes the statistical expectation. The distribution of a
circularly symmetric complex Gaussian (CSCG) random vector with mean
$\mv{x}$ and covariance matrix $\mv{\Sigma}$ is denoted by
$\mathcal{CN}(\mv{x},\mv{\Sigma})$, and $\sim$ stands for
``distributed as''. $\mathbb{C}^{x \times y}$ denotes the space of
$x\times y$ matrices with complex entries. $\|\mv{z}\|$ is the
Euclidean norm of a complex vector $\mv{z}$, and $|z|$ is the
absolute value of a complex scalar $z$. $\max(x,y)$ and $\min(x,y)$
denote the maximum and minimum between two real numbers, $x$ and
$y$, respectively, and $(x)^+=\max(x,0)$. All the $\log(\cdot)$
functions have base-2 by default.

\section{System Model And Problem Formulation}\label{sec:system model}

As shown in Fig. \ref{fig:system model}, this paper considers a
wireless broadcast system consisting of one transmitter, one EH
receiver, and one ID receiver. It is assumed that the transmitter is
equipped with $M\geq 1$ transmitting antennas, and the EH receiver
and the ID receiver are equipped with $N_{\rm EH}\geq 1$ and $N_{\rm
ID}\geq 1$ receiving antennas, respectively. In addition, it is
assumed that the transmitter and both receivers operate over the
same frequency band. Assuming a narrow-band transmission over
quasi-static fading channels, the baseband equivalent channels from
the transmitter to the EH receiver and ID receiver can be modeled by
matrices $\mv{G}\in\mathbb{C}^{N_{\rm EH}\times M}$ and
$\mv{H}\in\mathbb{C}^{N_{\rm ID}\times M}$, respectively. It is
assumed that at each fading state, $\mv{G}$ and $\mv{H}$ are both
known at the transmitter, and separately known at the corresponding
receiver. Note that for the case of co-located EH and ID receivers,
$\mv{G}$ is identical to $\mv{H}$ and thus $N_{\rm EH}=N_{\rm ID}$.

It is worth noting that the EH receiver does not need to convert the
received signal from the RF band to the baseband in order to harvest
the carried energy. Nevertheless, thanks to the law of energy
conservation, it can be assumed that the total harvested RF-band
power (energy normalized by the baseband symbol period), denoted by
$Q$, from all receiving antennas at the EH receiver is proportional
to that of the received baseband signal, i.e.,
\begin{align}\label{eq:power}
Q=\zeta\mathbb{E}[\|\mv{G}\mv{x}(n)\|^2]
\end{align}
where $\zeta$ is a constant that accounts for the loss in the energy
transducer for converting the harvested energy to electrical energy
to be stored; for the convenience of analysis, it is assumed that
$\zeta=1$ in this paper unless stated otherwise. We use
$\mv{x}(n)\in\mathbb{C}^{M\times 1}$ to denote the baseband signal
broadcast by the transmitter at the $n$th symbol interval, which is
assumed to be random over $n$, without loss of generality. The
expectation in (\ref{eq:power}) is thus used to compute the average
power harvested by the EH receiver at each fading state. Note that
for simplicity, we assumed in (\ref{eq:power}) that the harvested
energy due to the background noise at the EH receiver is negligible
and thus can be ignored.\footnote{The results of this paper are
readily extendible to study the impacts of non-negligible background
noise and/or co-channel interference on the SWIPT system
performance.}

On the other hand, the baseband transmission from the transmitter to
the ID receiver can be modeled by
\begin{align}\label{eq:signal model}
\mv{y}(n)=\mv{H}\mv{x}(n)+\mv{z}(n)
\end{align}
where $\mv{y}(n)\in\mathbb{C}^{N_{\rm ID}\times 1}$ denotes the
received signal at the $n$th symbol interval, and
$\mv{z}(n)\in\mathbb{C}^{N_{\rm ID} \times 1}$ denotes the receiver
noise vector. It is assumed that $\mv{z}(n)$'s are independent over
$n$ and $\mv{z}(n)\sim\mathcal{CN}(\mv{0},\mv{I})$. Under the
assumption that $\mv{x}(n)$ is random over $n$, we use
$\mv{S}=\mathbb{E}[\mv{x}(n)\mv{x}^H(n)]$ to denote the covariance
matrix of $\mv{x}(n)$. In addition, we assume that there is an
average power constraint at the transmitter across all transmitting
antennas denoted by
$\mathbb{E}[\|\mv{x}(n)\|^2]=\mathtt{tr}(\mv{S})\leq P$. In the
following, we examine the optimal transmit covariance $\mv{S}$ to
maximize the transported energy efficiency and information rate to
the EH and ID receivers, respectively.

Consider first the MIMO link from the transmitter to the EH receiver
when the ID receiver is not present. In this case, the design
objective for $\mv{S}$ is to maximize the power $Q$ received at the
EH receiver. Since from (\ref{eq:power}) it follows that
$Q=\mathtt{tr}(\mv{G}\mv{S}\mv{G}^H)$ with $\zeta=1$, the
aforementioned design problem can be formulated as
\begin{align}
\mbox{(P1)}~~\mathop{\mathtt{max}}_{\smv{S}} & ~~~
Q:=\mathtt{tr}\left(\mv{G}\mv{S}\mv{G}^H\right)
\nonumber \\
\mathtt{s.t.} & ~~~ \mathtt{tr}(\mv{S})\leq P, \mv{S}\succeq 0.
\nonumber
\end{align}
Let $T_1=\min(M,N_{\rm EH})$ and the (reduced) singular value
decomposition (SVD) of $\mv{G}$ be denoted by
$\mv{G}=\mv{U}_G\mv{\Gamma}_G^{1/2}\mv{V}_G^H$, where
$\mv{U}_G\in\mathbb{C}^{N_{\rm EH}\times T_1}$ and
$\mv{V}_G\in\mathbb{C}^{M\times T_1}$, each of which consists of
orthogonal columns with unit norm, and
$\mv{\Gamma}_G=\mathtt{diag}(g_1,\ldots,g_{T_1})$ with $g_1\geq
g_2\geq \ldots \geq g_{T_1}\geq 0$. Furthermore, let $\mv{v}_1$
denote the first column of $\mv{V}_G$. Then, we have the following
proposition.
\begin{proposition}\label{proposition:opt S EH}
The optimal solution to (P1) is $\mv{S}_{\rm
EH}=P\mv{v}_1\mv{v}_1^H$.
\end{proposition}
\begin{proof}
See Appendix \ref{appendix:proof opt S EH}.
\end{proof}

Given $\mv{S}=\mv{S}_{\rm EH}$, it follows that the maximum
harvested power at the EH receiver is given by $Q_{\max}=g_1P$. It
is worth noting that since $\mv{S}_{\rm EH}$ is a rank-one matrix,
the maximum harvested power is achieved by {\it beamforming} at the
transmitter, which aligns with the strongest eigenmode of the matrix
$\mv{G}^H\mv{G}$, i.e., the transmitted signal can be written as
$\mv{x}(n)=\sqrt{P}\mv{v}_1s(n)$, where $s(n)$ is an arbitrary
random signal over $n$ with zero mean and unit variance, and
$\mv{v}_1$ is the transmit beamforming vector. For convenience, we
name the above transmit beamforming scheme to maximize the
efficiency of WPT as ``energy beamforming''.

Next, consider the MIMO link from the transmitter to the ID receiver
without the presence of any EH receiver. Assuming the optimal
Gaussian codebook at the transmitter, i.e.,
$\mv{x}(n)\sim\mathcal{CN}(\mv{0},\mv{S})$, the transmit covariance
$\mv{S}$ to maximize the transmission rate over this MIMO channel
can be obtained by solving the following problem \cite{Cover}:
\begin{align}
\mbox{(P2)}~~\mathop{\mathtt{max}}_{\smv{S}} & ~~~ R:=
\log|\mv{I}+\mv{H}\mv{S}\mv{H}^H |
\nonumber \\
\mathtt{s.t.} & ~~~ \mathtt{tr}(\mv{S})\leq P, \mv{S}\succeq 0.
\nonumber
\end{align}
The optimal solution to the above problem is known to have the
following form \cite{Cover}: $\mv{S}_{\rm
ID}=\mv{V}_H\mv{\Lambda}\mv{V}_H^H$, where
$\mv{V}_H\in\mathbb{C}^{M\times T_2}$ is obtained from the (reduced)
SVD of $\mv{H}$ expressed by
$\mv{H}=\mv{U}_H\mv{\Gamma}_H^{1/2}\mv{V}_H^H$, with
$T_2=\min(M,N_{\rm ID})$, $\mv{U}_H\in\mathbb{C}^{N_{\rm ID}\times
T_2}$, $\mv{\Gamma}_H=\mathtt{ diag}(h_1,\ldots,h_{T_2})$, $h_1\geq
h_2\geq \ldots \geq h_{T_2}\geq 0$, and
$\mv{\Lambda}=\mathtt{diag}(p_1,\ldots,p_{T_2})$ with the diagonal
elements obtained from the standard ``water-filling (WF)'' power
allocation solution \cite{Cover}:
\begin{equation}\label{eq:WF}
p_i=\left(\nu-\frac{1}{h_i}\right)^+, \ \ i=1,\ldots,T_2
\end{equation}
with $\nu$ being the so-called (constant) water-level that makes
$\sum_{i=1}^{T_2}p_i=P$. The corresponding maximum transmission rate
is then given by $R_{\max}=\sum_{i=1}^{T_2}\log(1+h_ip_i)$. The
maximum rate is achieved in general by {\it spatial multiplexing}
\cite{Telater} over up to $T_2$ spatially decoupled AWGN channels,
together with the Gaussian codebook, i.e., the transmitted signal
can be expressed as $\mv{x}(n)=\mv{V}_H\mv{\Lambda}^{1/2}\mv{s}(n)$,
where $\mv{s}(n)$ is a Gaussian random vector
$\sim\mathcal{CN}(\mv{0},\mv{I})$, $\mv{V}_H$ and
$\mv{\Lambda}^{1/2}$ denote the precoding matrix and the (diagonal)
power allocation matrix, respectively.

\begin{remark}\label{remark:power rate relationship}
It is worth noting that in Problem (P1), it is assumed that the
transmitter sends to the EH receiver continuously. Now suppose that
the transmitter only transmits a fraction of the total time denoted
by $\alpha$ with $0<\alpha\leq 1$. Furthermore, assume that the
transmit power level can be adjusted flexibly provided that the
consumed average power is bounded by $P$, i.e., $\alpha \cdot
\mathtt{tr}(\mv{S})+(1-\alpha)\cdot 0 \leq P$ or
$\mathtt{tr}(\mv{S})\leq P/\alpha$. In this case, it can be easily
shown that the transmit covariance
$\mv{S}=(P/\alpha)\mv{v}_1\mv{v}_1^H$ also achieves the maximum
harvested power $Q_{\max}=g_1P$ for any $0<\alpha\leq 1$, which
suggests that the maximum power delivered is independent of
transmission time. However, unlike the case of maximum power
transfer, the maximum information rate reliably transmitted to the
ID receiver requires that the transmitter send signals continuously,
i.e., $\alpha=1$, as assumed in Problem (P2). This can be easily
verified by observing that for any $0<\alpha\leq 1$ and
$\mv{S}\succeq 0$, $\alpha \log|\mv{I}+\mv{H}(\mv{S}/\alpha)\mv{H}^H
|\leq \log|\mv{I}+\mv{H}\mv{S}\mv{H}^H |$ where the equality holds
only when $\alpha=1$, since $R$ is a nonlinear concave function of
$\mv{S}$. Thus, to maximize both power and rate transfer at the same
time, the transmitter should broadcast to the EH and ID receivers
all the time. Furthermore, note that the assumed Gaussian
distribution for transmitted signals is necessary for achieving the
maximum rate transfer, but not necessary for the maximum power
transfer. In fact, for any arbitrary complex number $c$ that
satisfies $|c|=1$, even a deterministic transmitted signal
$\mv{x}(n)=\sqrt{P}\mv{v}_1c, \forall n$, achieves the maximum
transferred power $Q_{\max}$ in Problem (P1). However, to maximize
simultaneous power and information transfer with the same
transmitted signal, the Gaussian input distribution is sufficient as
well as necessary.
\end{remark}

\begin{figure}
\centering{
 \epsfxsize=5in
    \leavevmode{\epsfbox{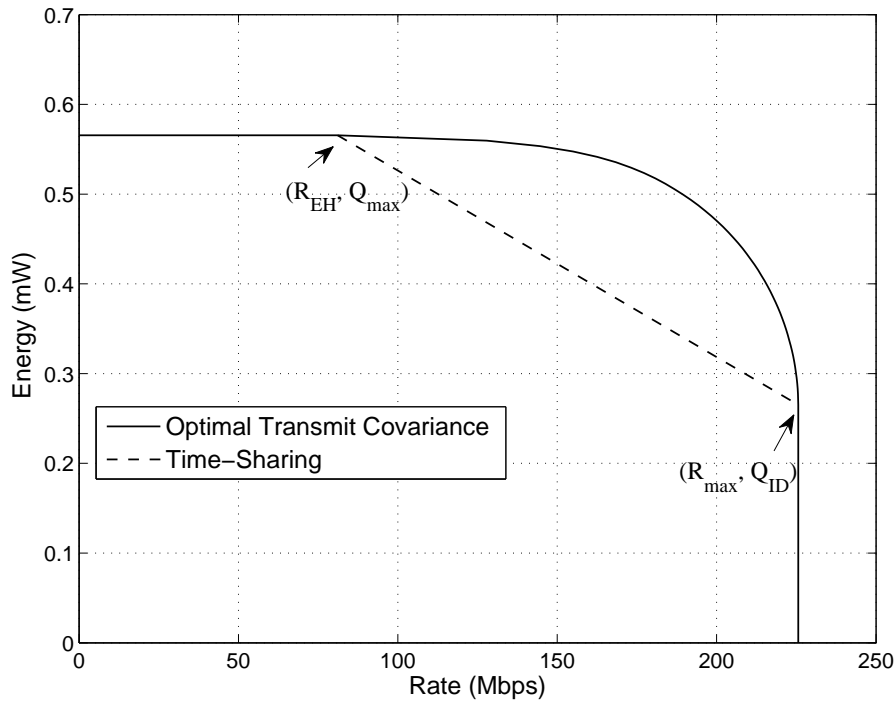}}}
\caption{Rate-energy tradeoff for a MIMO broadcast system with
separated EH and ID receivers, and $M=N_{\rm EH}=N_{\rm
ID}=4$.}\label{fig:RE region MIMO}
\end{figure}

Now, consider the case where both the EH and ID receivers are
present. From the above results, it is seen that the optimal
transmission strategies for maximal power transfer and information
transfer are in general different, which are energy beamforming and
information spatial multiplexing, respectively. It thus motivates
our investigation of the following question: What is the optimal
broadcasting strategy for simultaneous wireless power and
information transfer? To answer this question, we propose to use the
{\it Rate-Energy} (R-E) region (defined below) to characterize all
the achievable rate (in bits/sec/Hz or bps for information transfer)
and energy (in joule/sec or watt for power transfer) pairs under a
given transmit power constraint. Without loss of generality,
assuming that the transmitter sends Gaussian signals continuously
(cf. Remark \ref{remark:power rate relationship}), the R-E region is
defined as
\begin{align}\label{eq:RE region}
\mathcal{C}_{\rm R-E}(P)\triangleq\bigg\{(R,Q): R\leq
\log|\mv{I}+\mv{H}\mv{S}\mv{H}^H |, Q\leq
\mathtt{tr}(\mv{G}\mv{S}\mv{G}^H), \mathtt{tr}(\mv{S})\leq P,
\mv{S}\succeq 0 \bigg\}.
\end{align}

In Fig. \ref{fig:RE region MIMO}, an example of the above defined
R-E region (see Section \ref{sec:separated receivers} for the
algorithm to compute the boundary of this region) is shown for a
practical MIMO broadcast system with separated EH and ID receivers
(i.e., $\mv{G}\neq\mv{H}$). It is assumed that $M=N_{\rm EH}=N_{\rm
ID}=4$. The transmitter power is assumed to be $P=1$ watt(W) or
$30$dBm. The distances from the transmitter to the EH and ID
receivers are assumed to be $1$ meter and $10$ meters, respectively;
thus, we can exploit the near-far based energy and information
transmission scheduling, which may correspond to, e.g., a dedicated
energy transfer system (to ``near'' users) with opportunistic
information transmission (to ``far'' users), or vice versa. Assuming
a carrier frequency of $f_c=900$MHz and the power pathloss exponent
to be 4, the distance-dependent signal attenuation from the AP to
EH/ID receiver can be estimated as 40dB and 80dB, respectively.
Accordingly, the average signal power at the EH/ID receiver is thus
30dBm$-40$dB$=-10$dBm and 30dBm$-80$dB$=-50$dBm, respectively. It is
further assumed that in addition to signal pathloss, Rayleigh fading
is present, as such each element of channel matrices $\mv{G}$ and
$\mv{H}$ is independently drawn from the CSCG distribution with zero
mean and variance $-10$dBm (for EH receiver) and $-50$dBm for (for
ID receiver), respectively (to be consistent with the signal
pathloss previously assumed). Furthermore, the bandwidth of the
transmitted signal is assumed to be 10MHz, while the receiver noise
is assumed to be white Gaussian with power spectral density
$-140$dBm/Hz (which is dominated by the receiver processing noise
rather than the background thermal noise) or average power $-70$dBm
over the bandwidth of 10MHz. As a result, considering all of
transmit power, signal attenuation, fading and receiver noise, the
per-antenna average SNR at the ID receiver is equal to
$30-80-(-70)=20$dB, which corresponds to $P=100$ in the equivalent
signal model for the ID receiver given in (\ref{eq:signal model})
with unit-norm noise. In addition, we assume that for the EH
receiver, the energy conversion efficiency is $\zeta=$50\%.
Considering this together with transmit power and signal
attenuation, the average per-antenna signal power at the EH receiver
is thus $0.5\times$(30dBm$-40$dB) $=50$$\mu$W.

From Fig. \ref{fig:RE region MIMO}, it is observed that with energy
beamforming, the maximum harvested energy rate for the EH receiver
is around $Q_{\max}=0.57$mW, while with spatial multiplexing, the
maximum information rate for the ID receiver is around
$R_{\max}=225$Mbps. It is easy to identify two boundary points of
this R-E region denoted by $(R_{\rm EH}, Q_{\max})$ and $(R_{\max},
Q_{\rm ID})$, respectively. For the former boundary point, the
transmit covariance is $\mv{S}_{\rm EH}$, which corresponds to
transmit beamforming and achieves the maximum transferred power
$Q_{\max}$ to the EH receiver, while the resulting information rate
for the ID receiver is given by $R_{\rm
EH}=\log(1+\|\mv{H}\mv{v}_1\|^2P)$. On the other hand, for the
latter boundary point, the transmit covariance is $\mv{S}_{\rm ID}$,
which corresponds to transmit spatial multiplexing and achieves the
maximum information rate transferred to the ID receiver $R_{\max}$,
while the resulting power transferred  to the EH receiver is given
by $Q_{\rm ID}=\mathtt{tr}(\mv{G}\mv{S}_{\rm ID}\mv{G}^H)$.

Since the optimal tradeoff between the maximum energy and
information transfer rates is characterized by the boundary of the
R-E region, it is important to characterize all the boundary
rate-power pairs of $\mathcal{C}_{\rm R-E}(P)$ for any $P>0$. From
Fig. \ref{fig:RE region MIMO}, it is easy to observe that if $R\leq
R_{\rm EH}$, the maximum harvested power $Q_{\max}$ is achievable
with the same transmit covariance that achieves the rate-power pair
$(R_{\rm EH}, Q_{\max})$; similarly, the maximum information rate
$R_{\max}$ is achievable provided that $Q\leq Q_{\rm ID}$. Thus, the
remaining boundary of $\mathcal{C}_{\rm R-E}(P)$ yet to be
characterized is over the intervals: $R_{\rm EH}<R< R_{\max}$,
$Q_{\rm ID}<Q<Q_{\max}$. We thus consider the following optimization
problem:
\begin{align}
\mbox{(P3)}~~\mathop {\mathtt{max}}_{\smv{S}} & ~~~
\log\left|\mv{I}+\mv{H}\mv{S}\mv{H}^H\right|
\nonumber \\
\mathtt{s.t.} & ~~~ \mathtt{tr}\left(\mv{G}\mv{S}\mv{G}^H\right)\geq
\bar{Q}, ~\mathtt{tr}(\mv{S})\leq P, ~\mv{S}\succeq 0. \nonumber
\end{align}
Note that if $\bar{Q}$ takes values from $Q_{\rm
ID}<\bar{Q}<Q_{\max}$, the corresponding optimal rate solutions of
the above problems are the boundary rate points of the R-E region
over $R_{\rm EH}<R< R_{\max}$. Notice that the transmit covariance
solutions to the above problems in general yield larger rate-power
pairs than those by simply ``time-sharing'' the optimal transmit
covariance matrices $\mv{S}_{\rm EH}$ and $\mv{S}_{\rm ID}$ for EH
and ID receivers separately (see the dashed line in Fig. \ref{fig:RE
region MIMO}).\footnote{By time-sharing, we mean that the AP
transmits simultaneously to both EH and ID receivers with the
energy-maximizing transmit covariance $\mv{S}_{\rm EH}$ (i.e. energy
beamforming) for $\beta$ portion of each block time, and the
information-rate-maximizing transmit covariance $\mv{S}_{\rm ID}$
(i.e. spatial multiplexing) for the remaining $1-\beta$ portion of
each block time, with $0\leq\beta\leq 1$.}

Problem (P3) is a convex optimization problem, since its objective
function is concave over $\mv{S}$ and its constraints specify a
convex set of $\mv{S}$. Note that (P3) resembles a similar problem
formulated in \cite{Gastpar}, \cite{Zhang} (see also \cite{ZhangCR}
and references therein) under the cognitive radio (CR) setup, where
the rate of a secondary MIMO link is maximized subject to a set of
so-called {\it interference power constraints} to protect the
co-channel primary receivers. However, there is a key difference
between (P3) and the problem in \cite{Zhang}: the harvested power
constraint in (P3) has the reversed inequality of that of the
interference power constraint in \cite{Zhang}, since in our case it
is desirable for the EH receiver to harvest more power from the
transmitter, as opposed to that in \cite{Zhang} the interference
power at the primary receiver should be minimized. As such, it is
not immediately clear whether the solution in \cite{Zhang} can be
directly applied for solving (P3) with the reversed power
inequality. In the following, we will examine the solutions to
Problem (P3) for the two cases with arbitrary $\mv{G}$ and $\mv{H}$
(the case of separated receivers) and $\mv{G}=\mv{H}$ (the case of
co-located receivers), respectively.

\section{Separated Receivers}\label{sec:separated
receivers}

Consider the case where the EH receiver and ID receiver are
spatially separated and thus in general have different channels from
the transmitter. In this section, we first solve Problem (P3) with
arbitrary $\mv{G}$ and $\mv{H}$ and derive a semi-closed-form
expression for the optimal transmit covariance. Then, we examine the
optimal solution for the special case of MISO channels from the
transmitter to ID and/or EH receivers.

Since Problem (P3) is convex and satisfies the Slater's condition
\cite{Boyd}, it has a zero duality gap and thus can be solved using
the Lagrange duality method.\footnote{It is worth noting that
Problem (P3) is convex and thus can be solved efficiently by the
interior point method \cite{Boyd}; in this paper, we apply the
Lagrange duality method for this problem mainly to reveal the
optimal precoder structure.} Thus, we introduce two non-negative
dual variables, $\lambda$ and $\mu$, associated with the harvested
power constraint and transmit power constraint in (P3),
respectively. The optimal solution to Problem (P3) is then given by
the following theorem in terms of $\lambda^*$ and $\mu^*$, which are
the optimal dual solutions of Problem (P3) (see Appendix
\ref{appendix:proof optimal S} for details). Note that for Problem
(P3), given any pair of $\bar{Q}$ ($Q_{\rm ID}<\bar{Q}<Q_{\max}$)
and $P>0$, there exists one unique pair of $\lambda^*>0$ and
$\mu^*>0$.
\begin{theorem}\label{theorem:optimal S}
The optimal solution to Problem (P3) has the following form:
\begin{align}\label{eq:optimal S}
\mv{S}^*=\mv{A}^{-1/2}\tilde{\mv{V}}\tilde{\mv{\Lambda}}\tilde{\mv{V}}^H\mv{A}^{-1/2}
\end{align}
where $\mv{A}=\mu^*\mv{I}-\lambda^*\mv{G}^H\mv{G}$,
$\tilde{\mv{V}}\in\mathbb{C}^{M\times T_2}$ is obtained from the
(reduced) SVD of the matrix $\mv{H}\mv{A}^{-1/2}$ given by
$\mv{H}\mv{A}^{-1/2}=\tilde{\mv{U}}\tilde{\mv{\Gamma}}^{1/2}\tilde{\mv{V}}^H$,
with
$\tilde{\mv{\Gamma}}=\mathtt{diag}(\tilde{h}_1,\ldots,\tilde{h}_{T_2})$,
$\tilde{h}_1\geq \tilde{h}_2\geq \ldots \geq \tilde{h}_{T_2}\geq 0$,
and
$\tilde{\mv{\Lambda}}=\mathtt{diag}(\tilde{p}_1,\ldots,\tilde{p}_{T_2})$,
with $\tilde{p}_i=(1-1/\tilde{h}_i)^+, i=1,\ldots,T_2$.
\end{theorem}
\begin{proof}
See Appendix \ref{appendix:proof optimal S}.
\end{proof}
Note that this theorem requires that
$\mv{A}=\mu^*\mv{I}-\lambda^*\mv{G}^H\mv{G}\succ 0$, implying that
$\mu^*>\lambda^*g_1$ (recall that $g_1$ is the largest eigenvalue of
matrix $\mv{G}^H\mv{G}$), which is not present for a similar result
in \cite{ZhangCR} under the CR setup with the reversed interference
power constraint. One algorithm that can be used to solve (P3) is
provided in Table \ref{table} of Appendix \ref{appendix:proof
optimal S}. From Theorem \ref{theorem:optimal S}, the maximum
transmission rate for Problem (P3) can be shown to be
$R^*=\log\left|\mv{I}+\mv{H}\mv{S}^*\mv{H}^H\right|
=\sum_{i=1}^{T_2}\log(1+\tilde{h}_i\tilde{p}_i)$, for which the
proof is omitted here for brevity.

Next, we examine the optimal solution to Problem (P3) for the
special case where the ID receiver has one single antenna, i.e.,
$N_{\rm ID}=1$, and thus the MIMO channel $\mv{H}$ reduces to a row
vector $\mv{h}^H$ with $\mv{h}\in\mathbb{C}^{M\times 1}$. Suppose
that the EH receiver is still equipped with $N_{\rm EH}\geq 1$
antennas, and thus the MIMO channel $\mv{G}$ remains unchanged. From
Theorem \ref{theorem:optimal S}, we obtain the following corollary.
\begin{corollary}\label{corollary:optimal S MISO}
In the case of MISO channel from the transmitter to ID receiver,
i.e., $\mv{H}\equiv\mv{h}^H$, the optimal solution to Problem (P3)
reduces to the following form:
\begin{align}\label{eq:optimal S MISO}
\mv{S}^*=\mv{A}^{-1}\mv{h}\left(\frac{1}{\|\mv{A}^{-1/2}\mv{h}\|^2}-\frac{1}{\|\mv{A}^{-1/2}\mv{h}\|^4}\right)^+\mv{h}^H\mv{A}^{-1}
\end{align}
where $\mv{A}=\mu^*\mv{I}-\lambda^*\mv{G}^H\mv{G}$, with $\lambda^*$
and $\mu^*$ denoting the optimal dual solutions of Problem (P3).
Correspondingly, the optimal value of (P3) is
$R^*=\left(2\log\left(\|\mv{A}^{-1/2}\mv{h}\|\right)\right)^+$.
\end{corollary}
\begin{proof}
See Appendix \ref{appendix:proof MISO optimal S}.
\end{proof}

From (\ref{eq:optimal S MISO}), it is observed that the optimal
transmit covariance is a {\it rank-one} matrix, from which it
follows that {\it beamforming} is the optimal transmission strategy
in this case, where the transmit beamforming vector should be
aligned with the vector $\mv{A}^{-1}\mv{h}$. Moreover, consider the
case where both channels from the transmitter to ID/EH receivers are
MISO, i.e., $\mv{H}\equiv\mv{h}^H$, and $\mv{G}\equiv\mv{g}^H$ with
$\mv{g}\in\mathbb{C}^{M\times 1}$. From Corollary
\ref{corollary:optimal S MISO}, it follows immediately that the
optimal covariance solution to Problem (P3) is still beamforming. In
the following theorem, we show a closed-form solution of the optimal
beamforming vector at the transmitter for this special case, which
differs from the semi-closed-form solution (\ref{eq:optimal S MISO})
that was expressed in terms of dual variables.
\begin{theorem}\label{theorem:optimal S MISO new}
In the case of MISO channels from transmitter to both ID and EH
receivers, i.e., $\mv{H}\equiv\mv{h}^H$, and $\mv{G}\equiv\mv{g}^H$,
the optimal solution to Problem (P3) can be expressed as
$\mv{S}^*=P\mv{v}\mv{v}^H$, where the beamforming vector $\mv{v}$
has a unit-norm and is given by
\begin{eqnarray}\label{eq:optimal S MISO new}
\mv{v}=\left\{\begin{array}{ll} \hat{\mv{h}} &
0\leq\bar{Q}\leq |\mv{g}^H\hat{\mv{h}}|^2P \\
\sqrt{\frac{\bar{Q}}{P\|\smv{g}\|^2}}e^{j\angle\alpha_{gh}}\hat{\mv{g}}
+\sqrt{1-\frac{\bar{Q}}{P\|\smv{g}\|^2}}\hat{\mv{h}}_{g^\bot} &
|\mv{g}^H\hat{\mv{h}}|^2P<\bar{Q}\leq
P\|\mv{g}\|^2\end{array}\right.
\end{eqnarray}
where $\hat{\mv{h}}=\mv{h}/\|\mv{h}\|$,
$\hat{\mv{g}}=\mv{g}/\|\mv{g}\|$,
$\hat{\mv{h}}_{g^\bot}=\mv{h}_{g^\bot}/\|\mv{h}_{g^\bot}\|$ with
$\mv{h}_{g^\bot}=\mv{h}-(\hat{\mv{g}}^H\mv{h})\hat{\mv{g}}$, and
$\alpha_{gh}=\hat{\mv{g}}^H\mv{h}$ with $\angle
\alpha_{gh}\in[0,2\pi)$ denoting the phase of complex number
$\alpha_{gh}$. Correspondingly, the optimal value of (P3) is given
by
\begin{align}\label{eq:optimal rate MISO new}
R^*=\left\{\begin{array}{ll} \log(1+\|\mv{h}\|^2P) &
0\leq\bar{Q}\leq |\mv{g}^H\hat{\mv{h}}|^2P \\
\log\left(1+\left(\sqrt{\frac{\bar{Q}}{\|\smv{g}\|^2}}|\alpha_{gh}|+
\sqrt{P-\frac{\bar{Q}}{\|\smv{g}\|^2}}\sqrt{\|\mv{h}\|^2-|\alpha_{gh}|^2}\right)^2\right)
& |\mv{g}^H\hat{\mv{h}}|^2P<\bar{Q}\leq
P\|\mv{g}\|^2.\end{array}\right.
\end{align}
\end{theorem}
\begin{proof}
The proof is similar to that of Theorem 2 in \cite{Zhang}, and is
thus omitted for brevity.
\end{proof}

It is worth noting that in (\ref{eq:optimal S MISO new}), if
$\bar{Q}\leq |\mv{g}^H\hat{\mv{h}}|^2P$, the optimal transmit
beamforming vector is based on the principle of
maximal-ratio-combining (MRC) with respect to the MISO channel
$\mv{h}^H$ from the transmitter to the ID receiver, and in this
case, the harvested power constraint in Problem (P3) is indeed not
active; however, when $\bar{Q}> |\mv{g}^H\hat{\mv{h}}|^2P$, the
optimal beamforming vector is a linear combination of the two
vectors $\hat{\mv{g}}$ and $\hat{\mv{h}}_{g^\bot}$, and the
combining coefficients are designed such that the harvested power
constraint is satisfied with equality.

\begin{figure}
\centering{
 \epsfxsize=5in
    \leavevmode{\epsfbox{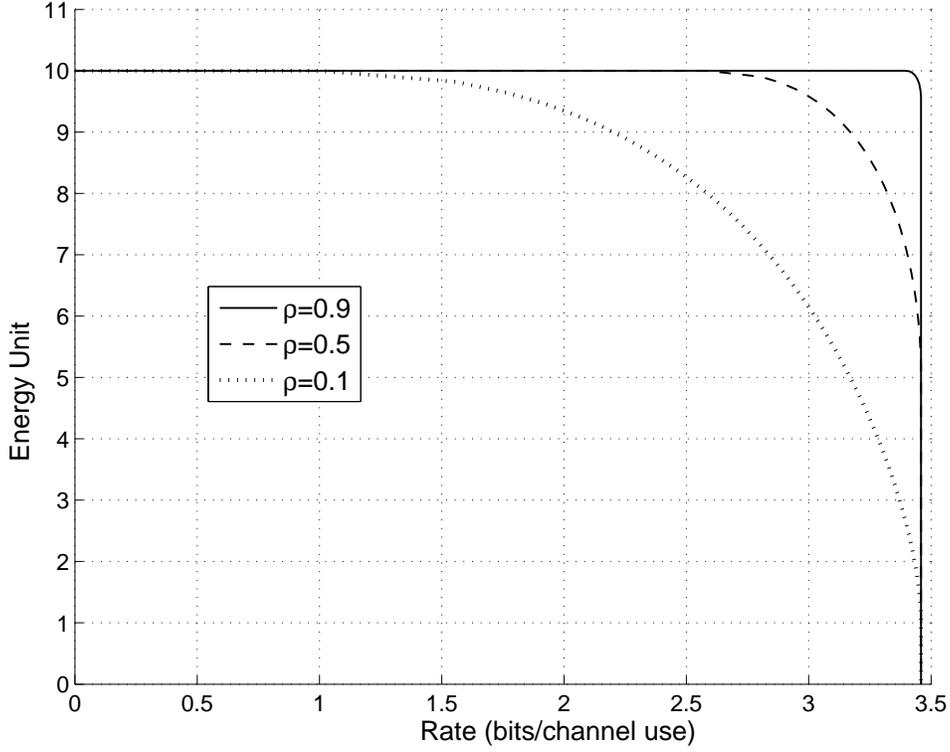}}}
\caption{Rate-energy tradeoff for a MISO broadcast system with
correlated MISO channels to the (separated) EH and ID
receivers.}\label{fig:RE region MISO}
\end{figure}

In Fig. \ref{fig:RE region MISO}, we show the achievable R-E regions
for the case of MISO channels from the transmitter to both EH and ID
receivers. We set $P=10$. For the purpose of exposition, it is
assumed that $\|\mv{h}\|=\|\mv{g}\|=1$ and $|\alpha_{gh}|^2=\rho$,
with $0\leq \rho \leq 1$ denoting the correlation between the two
unit-norm vectors $\mv{h}$ and $\mv{g}$. This channel setup may
correspond to the practical scenario where the EH and ID receivers
are equipped at a single device (but still physically separated),
and as a result their respective MISO channels from the transmitter
have the same power gain but are spatially correlated due to the
insufficient spacing between two separate receiving antennas. From
Theorem \ref{theorem:optimal S MISO new}, the R-E regions for the
three cases of $\rho=0.1,0.5$, and $0.9$ are obtained, as shown in
Fig. \ref{fig:RE region MISO}. Interestingly, it is observed that
increasing $\rho$ enlarges the achievable R-E region, which
indicates that the antenna correlation between the EH and ID
receivers can be a beneficial factor for simultaneous information
and power transfer. Note that in this figure, we express energy and
rate in terms of energy unit and bits/channel use, respectively,
since their practical values can be obtained by appropriate scaling
based on the realistic system parameters as for Fig. \ref{fig:RE
region MIMO}.

\section{Co-Located Receivers}\label{sec:colocated
receivers}

In this section, we address the case where the EH and ID receivers
are co-located, and thus possess the same channel from the
transmitter, i.e., $\mv{G}=\mv{H}$ and thus $N_{\rm EH}=N_{\rm
ID}\triangleq N$. We first examine the optimal solution of Problem
(P3) for this case, from which we obtain an outer bound for the
achievable rate-power pairs in the R-E region. Then, we propose two
practical receiver designs, namely time switching and power
splitting, derive their optimal transmission strategies to maximize
the achievable rate-power pairs, and finally compare the results to
the R-E region outer bound.

\subsection{Performance Outer Bound} \label{subsection:PB}

Consider Problem (P3) with $\mv{G}=\mv{H}$. Recall that the
(reduced) SVD of $\mv{H}$ is given by
$\mv{H}=\mv{U}_H\mv{\Gamma}_H^{1/2}\mv{V}_H^H$, with
$\mv{\Gamma}_H=\mathtt{ diag}(h_1,\ldots,h_{T_2})$, $h_1\geq h_2\geq
\ldots \geq h_{T_2}\geq 0$, and $T_2=\min(M,N)$. From Theorem
\ref{theorem:optimal S}, we obtain the following corollary.
\begin{corollary}\label{corollary:optimal S new}
In the case of co-located EH and ID receivers with $\mv{G}=\mv{H}$,
the optimal solution to Problem (P3) has the form of
$\mv{S}^*=\mv{V}_H\mv{\Sigma}\mv{V}_H^H$, where
$\mv{\Sigma}=\mathtt{diag}(\hat{p}_1,\ldots,\hat{p}_{T_2})$ with the
diagonal elements obtained from the following modified WF power
allocation:
\begin{equation}\label{eq:WF new}
\hat{p}_i=\left(\frac{1}{\mu^*-\lambda^*
h_i}-\frac{1}{h_i}\right)^+, \ \ i=1,\ldots,T_2
\end{equation}
with $\lambda^*$ and $\mu^*$ denoting the optimal dual solutions of
Problem (P3), $\mu^*>\lambda^*h_1$. The corresponding maximum
transmission rate is $R^*=\sum_{i=1}^{T_2}\log(1+h_i\hat{p}_i)$.
\end{corollary}
\begin{proof}
See Appendix \ref{appendix:proof optimal S co-located receivers}.
\end{proof}

The algorithm in Table \ref{table} for solving Problem (P3) with
arbitrary $\mv{G}$ and $\mv{H}$ can be simplified to solve the
special case with $\mv{G}=\mv{H}$. Corollary \ref{corollary:optimal
S new} reveals that for Problem (P3) in the case of $\mv{G}=\mv{H}$,
the optimal transmission strategy is in general still spatial
multiplexing over the eigenmodes of the MIMO channel $\mv{H}$ as for
Problem (P2), while the optimal tradeoffs between information
transfer and power transfer are achieved by varying the power levels
allocated into different eigenmodes, as shown in (\ref{eq:WF new}).
It is interesting to observe that the power allocation in
(\ref{eq:WF new}) reduces to the conventional WF solution in
(\ref{eq:WF}) with a constant water-level when $\lambda^*=0$, i.e.,
the harvested power constraint in Problem (P3) is inactive with the
optimal power allocation. However, when $\lambda^*>0$ and thus the
harvested power constraint is active corresponding to the
Pareto-optimal regime of our interest, the power allocation in
(\ref{eq:WF new}) is observed to have a {\it non-decreasing}
water-level as $h_i$'s increase. Note that this modified WF policy
has also been shown in \cite{Grover} for power allocation in
frequency-selective AWGN channels.

Using Corollary \ref{corollary:optimal S new}, we can characterize
all the boundary points of the R-E region $\mathcal{C}_{\rm R-E}(P)$
defined in (\ref{eq:RE region}) for the case of co-located receivers
with $\mv{G}=\mv{H}$. For example, if the total transmit power is
allocated to the channel with the largest gain $h_1$, i.e.,
$\hat{p}_1=P$ and $\hat{p}_i=0, i=2,\ldots,T_2$, the maximum
harvested power $Q_{\max}=Ph_1$ is achieved by transmit beamforming.
On the other hand, if transmit spatial multiplexing is applied with
the conventional WF power allocation given in (\ref{eq:WF new}) with
$\lambda^*=0$, the corresponding $R^*$ becomes the maximum
transmission rate, $R_{\max}$. However, unlike the case of separated
EH and ID receivers in which the entire boundary of
$\mathcal{C}_{\rm R-E}(P)$ is achievable, in the case of co-located
receivers, except the two boundary rate-power pairs $(R_{\max},0)$
and $(0, Q_{\max})$, all the other boundary pairs of
$\mathcal{C}_{\rm R-E}(P)$ may not be achievable in practice. Note
that these boundary points are achievable if and only if (iff) the
following premise is true: the power of the received signal across
all antennas is totally harvested, and at the same time the carried
information with a transmission rate up to the MIMO channel capacity
(for a given transmit covariance) is decodable. However, existing EH
circuits are not yet able to directly decode the information carried
in the RF-band signal, even for the SISO channel case; as a result,
how to achieve the remaining boundary rate-power pairs of
$\mathcal{C}_{\rm R-E}(P)$ in the MIMO case with the co-located EH
and ID receiver remains an interesting open problem. Therefore, in
the case of co-located receivers, the boundary of $\mathcal{C}_{\rm
R-E}(P)$ given by Corollary \ref{corollary:optimal S new} in general
only serves as an {\it outer bound} for the achievable rate-power
pairs with practical receiver designs, as will be investigated in
the following subsections.

\subsection{Time Switching}

First, as shown in Fig. \ref{fig:practical schemes}(a), we consider
the {\it time switching} (TS) scheme, with which each transmission
block is divided into two orthogonal time slots, one for
transferring power and the other for transmitting data. The
co-located EH and ID receiver switches its operations periodically
between harvesting energy and decoding information between the two
time slots. It is assumed that time synchronization has been
perfectly established between the transmitter and the receiver, and
thus the receiver can synchronize its function switching with the
transmitter. With orthogonal transmissions, the transmitted signals
for the EH receiver and ID receiver can be designed separately, but
subject to a total transmit power constraint. Let $\alpha$ with
$0\leq \alpha \leq 1$ denote the percentage of transmission time
allocated to the EH time slot. We then consider the following two
types of power constraints at the transmitter:

\begin{itemize}
\item {\it Fixed power constraint}: The transmitted signals to the ID and EH receivers have the same fixed power constraint given
by $\mathtt{tr}(\mv{S}_1)\leq P$, and $\mathtt{tr}(\mv{S}_2)\leq P$,
where $\mv{S}_1$ and $\mv{S}_2$ denote the transmit covariance
matrices for the ID and EH transmission time slots, respectively.

\item {\it Flexible power constraint}: The transmitted signals to the ID and EH receivers can have different power constraints provided that their average
consumed power is below $P$, i.e.,
$(1-\alpha)\mathtt{tr}(\mv{S}_1)+\alpha\mathtt{tr}(\mv{S}_2)\leq P$.
\end{itemize}

Note that the TS scheme under the fixed power constraint has been
considered in \cite{Varshney Thesis} for the single-antenna AWGN
channel. The achievable R-E regions for the TS scheme with the fixed
(referred to as ${\rm TS_1}$) vs. flexible (referred to as ${\rm
TS_2}$) power constraints are then given as follows:
\begin{align}\label{eq:RE region time switching 1} \mathcal{C}_{\rm
R-E}^{\rm TS_1}(P)\triangleq \bigcup_{0\leq\alpha\leq
1}\bigg\{(R,Q):
R\leq (1-\alpha)\log|\mv{I}+\mv{H}\mv{S}_1\mv{H}^H |, \nonumber \\
Q\leq
\alpha\mathtt{tr}(\mv{H}\mv{S}_2\mv{H}^H),\mathtt{tr}(\mv{S}_1)\leq
P, \mathtt{tr}(\mv{S}_2)\leq P \bigg\}
\end{align}
\begin{align}\label{eq:RE region time switching 2}
\mathcal{C}_{\rm R-E}^{\rm TS_2}(P)\triangleq
\bigcup_{0\leq\alpha\leq 1}\bigg\{(R,Q): R\leq
(1-\alpha)\log|\mv{I}+\mv{H}\mv{S}_1\mv{H}^H |, \nonumber \\ Q\leq
\alpha\mathtt{tr}(\mv{H}\mv{S}_2\mv{H}^H),
(1-\alpha)\mathtt{tr}(\mv{S}_1)+\alpha\mathtt{tr}(\mv{S}_2)\leq P
\bigg\}.
\end{align}

It is worth noting that $\mathcal{C}_{\rm R-E}^{\rm
TS_1}(P)\subseteq\mathcal{C}_{\rm R-E}^{\rm TS_2}(P)$ must be true
since any pair of $\mv{S}_1\succeq 0$ and $\mv{S}_2\succeq 0$ that
satisfy the fixed power constraint will satisfy the flexible power
constraint, but not vice versa. The optimal transmit covariance
matrices $\mv{S}_1$ and $\mv{S}_2$ to achieve the boundary of
$\mathcal{C}_{\rm R-E}^{\rm TS_1}(P)$ with the fixed power
constraint are given in Section \ref{sec:system model} (assuming
$\mv{G}=\mv{H}$). In fact, the boundary of $\mathcal{C}_{\rm
R-E}^{\rm TS_1}(P)$ is simply a straight line connecting the two
points $(R_{\max},0)$ and $(0, Q_{\max})$ (cf. Fig. \ref{fig:RE
region co-located receiver 1}) by sweeping $\alpha$ from 0 to 1.

Similarly, for the case of flexible power constraint, the transmit
covariance solutions for $\mv{S}_1$ and $\mv{S}_2$ to achieve any
boundary point of $\mathcal{C}_{\rm R-E}^{\rm TS_2}(P)$ can be shown
to have the same set of eigenvectors as those given in Section
\ref{sec:system model} (assuming $\mv{G}=\mv{H}$), respectively;
however, the corresponding time allocation for $\alpha$ and power
allocation for $\mv{S}_1$ and $\mv{S}_2$ remain unknown. We thus
have the following proposition.
\begin{proposition}\label{proposition:TS zero alpha}
In the case of flexible power constraint, except the two points
$(R_{\max},0)$ and $(0, Q_{\max})$, all other boundary points of the
region $\mathcal{C}_{\rm R-E}^{\rm TS_2}(P)$ are achieved as
$\alpha\rightarrow 0$; accordingly, $\mathcal{C}_{\rm R-E}^{\rm
TS_2}(P)$ can be simplified as
\begin{align}\label{eq:RE region time switching 2 new}
\mathcal{C}_{\rm R-E}^{\rm TS_2}(P)=\bigg\{(R, Q): R\leq
\log|\mv{I}+\mv{H}\mv{S}_1\mv{H}^H |, \mathtt{tr}(\mv{S}_1)\leq
(P-Q/h_1), \mv{S}_1\succeq 0 \bigg\}.\end{align}
\end{proposition}
\begin{proof}
See Appendix \ref{appendix:proof zero alpha}.
\end{proof}

The corresponding optimal power allocation for $\mv{S}_1$ and
$\mv{S}_2$ can be easily obtained given (\ref{eq:RE region time
switching 2 new}) and are thus omitted for brevity. Proposition
\ref{proposition:TS zero alpha} suggests that to achieve any
boundary point $(R,Q)$ of $\mathcal{C}_{\rm R-E}^{\rm TS_2}(P)$ with
$R < R_{\max}$ and $Q < Q_{\max}$, the portion of transmission time
$\alpha$ allocated to power transfer in each block should
asymptotically go to zero when $n\rightarrow \infty$, where $n$
denotes the number of transmitted symbols in each block. For
example, by allocating $O(\log n)$ symbols per block for power
transfer and the remaining symbols for information transmission
yields $\alpha=\log n/n\rightarrow 0$ as $n\rightarrow \infty$,
which satisfies the optimality condition given in Proposition
\ref{proposition:TS zero alpha}.

It is worth noting that the boundary of $\mathcal{C}_{\rm R-E}^{\rm
TS_2}(P)$ in the flexible power constraint case is achieved under
the assumption that the transmitter and receiver can both operate in
the regime of infinite power in the EH time slot due to
$\alpha\rightarrow 0$, which cannot be implemented with practical
power amplifiers. Hence, a more feasible region for
$\mathcal{C}_{\rm R-E}^{\rm TS_2}(P)$ is obtained by adding
peak\footnote{Note that the peak power constraint in this context is
different from the signal amplitude constraint considered in
\cite{Varshney}, \cite{Varshney Thesis}.} transmit power constraints
in (\ref{eq:RE region time switching 2}) as
$\mathtt{tr}(\mv{S}_1)\leq P_{\rm peak}$ and
$\mathtt{tr}(\mv{S}_2)\leq P_{\rm peak}$, with $P_{\rm peak}\geq P$.
Similar to Proposition \ref{proposition:TS zero alpha}, it can be
shown that the boundary of the achievable R-E region in this case,
denoted by $\mathcal{C}_{\rm R-E}^{\rm TS_2}(P,P_{\rm peak})$, is
achieved by $\alpha=Q/(h_1P_{\rm peak})$. Note that we can
equivalently denote the achievable R-E region $\mathcal{C}_{\rm
R-E}^{\rm TS_2}(P)$ defined in (\ref{eq:RE region time switching 2})
or (\ref{eq:RE region time switching 2 new}) without any peak power
constraint as $\mathcal{C}_{\rm R-E}^{\rm TS_2}(P,\infty)$.

\subsection{Power Splitting}

\begin{figure}
\begin{center}
\scalebox{0.6}{\includegraphics*[77pt,368pt][496pt,738pt]{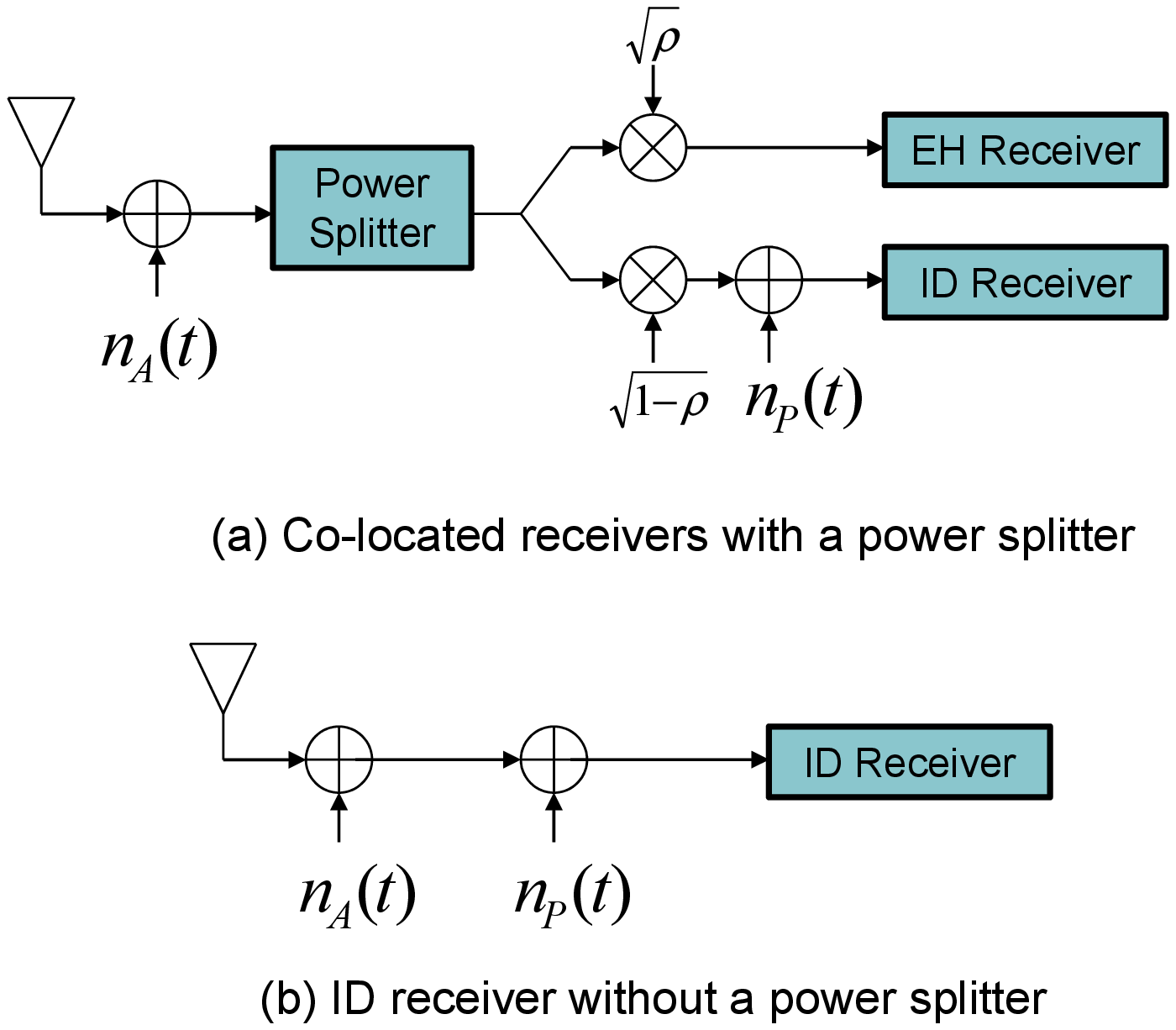}}
\end{center}
\caption{Receiver operations with/without a power splitter (the
energy harvested due to the receiver noise is ignored for EH
receiver).} \label{fig:PS}
\end{figure}

Next, we propose an alternative receiver design called {\it power
splitting} (PS), whereby the power and information transfer to the
co-located EH and ID receivers are simultaneously achieved via a set
of power splitting devices, one device for each receiving antenna,
as shown in Fig. \ref{fig:practical schemes}(b). In order to gain
more insight into the PS scheme, we consider first the simple case
of a single-antenna AWGN channel with co-located ID and EH
receivers, which is shown in Fig. \ref{fig:PS}(a). For the ease of
comparison, the case of solely information transfer with one single
ID receiver is also shown in Fig. \ref{fig:PS}(b).

The receiver operations in Fig. \ref{fig:PS}(a) are explained as
follows: The received signal from the antenna is first corrupted by
a Gaussian noise denoted by $n_A(t)$ at the RF-band, which is
assumed to have zero mean and equivalent baseband power
$\sigma_A^2$. The RF-band signal is then fed into a power splitter,
which is assumed to be perfect without any noise induced. After the
power splitter, the portion of signal power split to the EH receiver
is denoted by $\rho$, and that to the ID receiver by $1-\rho$. The
signal split to the ID receiver then goes through a sequence of
standard operations (see, e.g., \cite{Proakis}) to be converted from
the RF band to baseband. During this process, the signal is
additionally corrupted by another noise $n_P(t)$, which is
independent of $n_A(t)$ and assumed to be Gaussian and have zero
mean and variance $\sigma_P^2$. To be consistent to the case with
solely the ID receiver, it is reasonable to assume that the antenna
noise $n_A(t)$ and processing noise $n_P(t)$ have the same
distributions in both Figs. \ref{fig:PS}(a) and \ref{fig:PS}(b). It
is further assumed that $\sigma_A^2+\sigma_P^2=1$ to be consistent
with the system model introduced in Section \ref{sec:system model}.

For this simple SISO AWGN channel, we denote the transmit power
constraint by $P$ and the channel power gain by $h$. It is then easy
to compute the R-E region outer bound $\mathcal{C}_{\rm R-E}(P)$ for
this channel with co-located ID and EH receivers, which is simply a
box specified by three vertices $(0, Q_{\max})$, $(R_{\max},0)$ and
$(R_{\max}, Q_{\max})$, with $Q_{\max}=Ph$ and
$R_{\max}=\log(1+Ph)$. Interestingly, we will show next that under
certain conditions, the PS scheme can in fact achieve all the
rate-energy pairs in this R-E region outer bound; without loss of
generality, it suffices to show that the vertex point $(R_{\max},
Q_{\max})$ is achievable.

With reference to Fig. \ref{fig:PS}(a), we discuss the PS scheme in
the following three regimes with different values of antenna and
processing noise power.
\begin{itemize}
\item $\sigma_A^2\ll \sigma_P^2$ (Case I): In this ideal case with perfect receiving antenna,
the antenna noise can be ignored and thus we have $\sigma_A^2=0$ and
$\sigma_P^2=1$. Accordingly, it is easy to show that the SNR,
denoted by $\tau$, at the ID receiver in Fig. \ref{fig:PS}(a) is
$(1-\rho)Ph$. The achievable R-E region in this case is then given
by $\mathcal{C}_{\rm R-E}^{\rm PS, I}(P) \triangleq \bigcup_{0\leq
\rho\leq 1} \{(R,Q): R\leq \log(1+(1-\rho)Ph), Q\leq \rho Ph\}$.
This region can be shown to coincide with the R-E region for the TS
scheme with the flexible power constraint given by (\ref{eq:RE
region time switching 2 new}) for the SISO case.

\item $0<\sigma_A^2<1$ (Case II): This is the most practically valid case. Since $\sigma_P^2=1-\sigma_A^2$, we can show that $\tau$ in
this case is given by
$\tau=(1-\rho)Ph/((1-\rho)\sigma_A^2+\sigma_P^2)=(1-\rho)Ph/(1-\rho\sigma_A^2)$.
Accordingly, the achievable R-E region in this case is given by
$\mathcal{C}_{\rm R-E}^{\rm PS, II}(P) \triangleq \bigcup_{0\leq
\rho\leq 1} \{(R,Q): R\leq \log(1+\tau), Q\leq \rho Ph\}$. It is
easy to show that $\mathcal{C}_{\rm R-E}^{\rm PS, II}(P)$ enlarges
strictly as $\sigma_A^2$ increases from 0 to 1.

\item $\sigma_A^2\gg \sigma_P^2$ (Case III): In this ideal case with
perfect RF-to-baseband signal conversion, the processing noise can
be ignored and thus we have $\sigma_P^2=0$ and $\sigma_A^2=1$. In
this case, the SNR for the ID receiver is given by $\tau=Ph$,
regardless of the value of $\rho$. Thus, to maximize the power
transfer, ideally we should set $\rho\rightarrow 1$, i.e., splitting
infinitesimally small power to the ID receiver since both the signal
and antenna noise are scaled identically by the power splitter and
there is no additional processing noise induced after the power
splitting. With $\rho=1$, the achievable R-E region in this case is
given by $\mathcal{C}_{\rm R-E}^{\rm PS, III}(P) \triangleq \{(R,Q):
R\leq \log(1+Ph), Q\leq Ph\}$, which becomes identical to the R-E
region outer bound $\mathcal{C}_{\rm R-E}(P)$ (which is a box as
defined earlier).
\end{itemize}

Therefore, we know from the above discussions that only for the case
of noise-free RF-band to baseband processing (i.e., Case III), the
PS scheme achieves the R-E region outer bound and is thus optimal.
However, in practice, such a condition can never be met perfectly,
and thus the R-E region outer bound $\mathcal{C}_{\rm R-E}(P)$ is in
general still non-achievable with practical PS receivers. In the
following, we will study further the achievable R-E region by the PS
scheme for the more general case of MIMO channels. It is not
difficult to show that if each receiving antenna satisfies the
condition in Case III, the R-E region outer bound $\mathcal{C}_{\rm
R-E}(P)$ defined in (\ref{eq:RE region}) with $\mv{G}=\mv{H}$ is
achievable for the MIMO case by the PS scheme (with each receiving
antenna to set $\rho=1$). For a more practical purpose, we consider
in the rest of this section the ``worst'' case performance of the PS
scheme (i.e., Case I in the above), when the noiseless antenna is
assumed (which leads to the smallest R-E region for the SISO AWGN
channel case). The obtained R-E region will thus provide the
performance lower bound for the PS scheme with practical receiver
circuits. In this case, since there is no antenna noise and the
processing noise is added after the power splitting, it is
equivalent to assume that the aggregated receiver noise power
remains unchanged with a power splitter at each receiving antenna.
Let $\rho_i$ with $0\leq \rho_i\leq 1$ denote the portion of power
split to the EH receiver at the $i$th receiving antenna, $1\leq
i\leq N$. The achievable R-E region for the PS scheme (in the worst
case) is thus given by
\begin{align}\label{eq:RE region power splitting}
\mathcal{C}_{\rm R-E}^{\rm PS}(P) \triangleq \bigcup_{0\leq
\rho_i\leq 1, \forall i} \bigg\{(R,Q):   R\leq
\log|\mv{I}+\bar{\mv{\Lambda}}_{\rho}^{1/2}\mv{H}\mv{S}\mv{H}^H\bar{\mv{\Lambda}}_{\rho}^{1/2}|,
Q\leq \mathtt{tr}(\mv{\Lambda}_{\rho}\mv{H}\mv{S}\mv{H}^H),
\mathtt{tr}(\mv{S})\leq P, \mv{S}\succeq 0 \bigg\}
\end{align}
where $\mv{\Lambda}_{\rho}=\mathtt{diag}(\rho_1, \ldots, \rho_{N})$,
and $\bar{\mv{\Lambda}}_{\rho}=\mv{I}-\mv{\Lambda}_{\rho}$.

Note that the two points $(R_{\max},0)$ and $(0, Q_{\max})$ on the
boundary of $\mathcal{C}_{\rm R-E}^{\rm PS}(P)$ can be simply
achieved with $\rho_i=0, \forall i$, and $\rho_i=1, \forall i$,
respectively, with the corresponding transmit covariance matrices
given in Section \ref{sec:system model} (with $\mv{G}=\mv{H}$),
similar to the TS case. All the other boundary points of
$\mathcal{C}_{\rm R-E}^{\rm PS}(P)$ can be obtained as follows: Let
$\mv{H}'=\bar{\mv{\Lambda}}_{\rho}^{1/2}\mv{H}$,
$\mv{G}'=\mv{\Lambda}_{\rho}^{1/2}\mv{H}$, and $\mathcal{R}_{\rm
R-E}^{\rm PS}(P, \{\rho_i\})$ denote the achievable R-E region with
PS for a given set of $\rho_i$'s. Then, we can obtain the boundary
of $\mathcal{R}_{\rm R-E}^{\rm PS}(P, \{\rho_i\})$ by solving
similar problems like Problem (P3) (with $\mv{H}$ and $\mv{G}$
replaced by $\mv{H}'$ and $\mv{G}'$, respectively). Finally, the
boundary of $\mathcal{C}_{\rm R-E}^{\rm PS}(P)$ can be obtained by
taking a union operation over different $\mathcal{R}^{\rm
PS}(P,\{\rho_i\})$'s with all possible $\rho_i$'s.

In particular, we consider two special cases of the PS scheme: i)
{\it Uniform Power Splitting} (UPS) with $\rho_i=\rho, \forall i$,
and $0\leq \rho \leq 1$; and ii) {\it On-Off Power Splitting} with
$\rho_i\in \{0,1\}, \forall i$, i.e., $\rho_i$ taking the value of
either 0 or 1. For the case of on-off power splitting, let
$\Omega\subseteq\{1,\ldots,N\}$ denote one subset of receiving
antennas with $\rho_i=1$; then $\bar{\Omega}=\{1,\ldots,N\}-\Omega$
denotes the other subset of receiving antennas with $\rho_i=0$.
Clearly, $\Omega$ and $\bar{\Omega}$ specify the sets of receiving
antennas switched to EH and ID receivers, respectively; thus, the
on-off power splitting is also termed {\it Antenna Switching} (AS).

Let $\mathcal{R}^{\rm UPS}(P,\rho)$ denote the achievable R-E region
for the UPS scheme with any fixed $\rho$, and $\mathcal{C}_{\rm
R-E}^{\rm UPS}(P)$ be the R-E region by taking the union of all
$\mathcal{R}^{\rm UPS}(P,\rho)$'s over $0\leq \rho \leq 1$.
Furthermore, let $\mathcal{R}^{\rm AS}(P,\Omega)$ denote the
achievable R-E region for the AS (or On-Off Power Splitting) scheme
with a given pair of $\Omega$ and $\bar{\Omega}$. It is not
difficult to see that for any $P>0$, $\mathcal{C}_{\rm R-E}^{\rm
UPS}(P)\subseteq\mathcal{C}_{\rm R-E}^{\rm PS}(P)$, and
$\mathcal{R}^{\rm AS}(P,\Omega)\subseteq\mathcal{C}_{\rm R-E}^{\rm
PS}(P), \forall \Omega$, while $\mathcal{C}_{\rm R-E}^{\rm
UPS}(P)=\mathcal{C}_{\rm R-E}^{\rm PS}(P)$ for the case of MISO/SISO
channel of $\mv{H}$. Moreover, the following proposition shows that
for the case of SIMO channel of $\mv{H}$, $\mathcal{C}_{\rm
R-E}^{\rm UPS}(P)=\mathcal{C}_{\rm R-E}^{\rm PS}(P)$ is also true.

\begin{proposition}\label{proposition:SIMO}
In the case of co-located EH and ID receivers with a SIMO channel
$\mv{H}\equiv\mv{h}\in\mathbb{C}^{N\times 1}$, for any $P\geq 0$,
$\mathcal{C}_{\rm R-E}^{\rm UPS}(P)=\mathcal{C}_{\rm R-E}^{\rm
PS}(P)=\{(R,Q): R\leq \log(1+(\|\mv{h}\|^2P-Q)), 0\leq Q \leq
\|\mv{h}\|^2P\}$.
\end{proposition}
\begin{proof}
See Appendix \ref{appendix:SIMO}.
\end{proof}

\subsection{Performance Comparison}

The following proposition summarizes the performance comparison
between the TS and UPS schemes.
\begin{proposition}\label{proposition:TS UPS compare}
For the co-located EH and ID receivers, with any $P>0$,
$\mathcal{C}_{\rm R-E}^{\rm TS_1}(P)\subseteq\mathcal{C}_{\rm
R-E}^{\rm UPS}(P)\subseteq \mathcal{C}_{\rm R-E}^{\rm TS_2}(P)$,
while $\mathcal{C}_{\rm R-E}^{\rm UPS}(P)=\mathcal{C}_{\rm R-E}^{\rm
TS_2}(P)$ iff $P\leq (1/h_2-1/h_1)$.
\end{proposition}
\begin{proof}
See Appendix \ref{appendix:TS UPS compare}.
\end{proof}
From the above proposition, it follows that the TS scheme with the
fixed power constraint performs worse than the UPS scheme in terms
of achievable rate-energy pairs. However, the UPS scheme in general
performs worse than the TS scheme under the flexible power
constraint (without any peak power constraint), while they perform
identically iff the condition $P\leq (1/h_2-1/h_1)$ is satisfied.
This may occur when, e.g., $P$ is sufficiently small (unlikely in
our model since high SNR is of interest), or $h_2=0$ (i.e., $\mv{H}$
is MISO or SIMO). Note that the performance comparison between the
TS scheme (with the flexible power constraint) and the PS scheme
with arbitrary power splitting (instead of UPS) remains unknown
theoretically.

Next, for the purpose of exposition, we compare the rate-energy
tradeoff for the case of co-located EH and ID receivers for a
symmetric MIMO channel $\mv{G}=\mv{H}=[1, ~\theta; \theta, ~ 1]$
with $\theta=0.5$ for Fig. \ref{fig:RE region co-located receiver 1}
and $\theta=0.8$ for Fig. \ref{fig:RE region co-located receiver 2},
respectively. It is assumed that $P=100$. The R-E region outer bound
is obtained as $\mathcal{C}_{\rm R-E}(P)$ with $\mv{G}=\mv{H}$
according to Corollary \ref{corollary:optimal S new}. The two
achievable R-E regions for the TS scheme with fixed vs. flexible
power constraints are shown for comparison, and it is observed that
$\mathcal{C}_{\rm R-E}^{\rm TS_1}(P)\subseteq\mathcal{C}_{\rm
R-E}^{\rm TS_2}(P)$. The achievable R-E region for the TS scheme
with the flexible power constraint $P$ as well as the peak power
constraint $P_{\rm peak}=2P$ is also shown, which is observed to lie
between $\mathcal{C}_{\rm R-E}^{\rm TS_1}(P)$ and $\mathcal{C}_{\rm
R-E}^{\rm TS_2}(P)$. Moreover, the achievable R-E region
$\mathcal{C}_{\rm R-E}^{\rm UPS}(P)$ for the UPS scheme is shown,
whose boundary points constitute those of $\mathcal{R}^{\rm
UPS}(P,\rho)$'s with different $\rho$'s from 0 to 1. It is observed
that $\mathcal{C}_{\rm R-E}^{\rm TS_1}(P)\subseteq\mathcal{C}_{\rm
R-E}^{\rm UPS}(P)\subseteq \mathcal{C}_{\rm R-E}^{\rm TS_2}(P)$,
which is in accordance with Proposition \ref{proposition:TS UPS
compare}. Note that for this channel, the R-E region
$\mathcal{C}_{\rm R-E}^{\rm PS}(P)$ for the general PS scheme
defined in (\ref{eq:RE region power splitting}) only provides
negligible rate-energy gains over $\mathcal{C}_{\rm R-E}^{\rm
UPS}(P)$ by the UPS scheme, and is thus not shown here. In addition,
the achievable R-E region $\mathcal{R}^{\rm AS}(P,\Omega)$ for the
AS scheme is shown, which is the same for $\Omega=\{1\}$ or $\{2\}$
due to the symmetric channel setup. Furthermore, by comparing Figs.
\ref{fig:RE region co-located receiver 1} and \ref{fig:RE region
co-located receiver 2}, it is observed that the performance gap
between $\mathcal{C}_{\rm R-E}^{\rm UPS}(P)$ and $\mathcal{C}_{\rm
R-E}^{\rm TS_2}(P)$ is reduced when $\theta$ increases from $0.5$ to
$0.8$. This is because for this channel setup, $h_1=(1+\theta)^2$
and $h_2=(1-\theta)^2$, and as a result, the condition $P\leq
(1/h_2-1/h_1)$  in Proposition \ref{proposition:TS UPS compare} for
$\mathcal{C}_{\rm R-E}^{\rm UPS}(P)=\mathcal{C}_{\rm R-E}^{\rm
TS_2}(P)$ will hold when $\theta\rightarrow 1$ for any $P>0$.
Finally, it is worth pointing out that in practical SWIPT systems
with the co-located EH/ID receiver, under the high-SNR condition,
the receiver typically operates at the ``high-energy'' regime in the
achievable rate-energy regions shown in Figs. \ref{fig:RE region
co-located receiver 1} and \ref{fig:RE region co-located receiver
2}, which corresponds to applying very large values of the
time-switching coefficient $\alpha$ or the power-splitting
coefficient $\rho$, i.e, $\alpha\rightarrow 1$ and $\rho\rightarrow
1$.

\begin{figure}
\centering{
 \epsfxsize=5in
    \leavevmode{\epsfbox{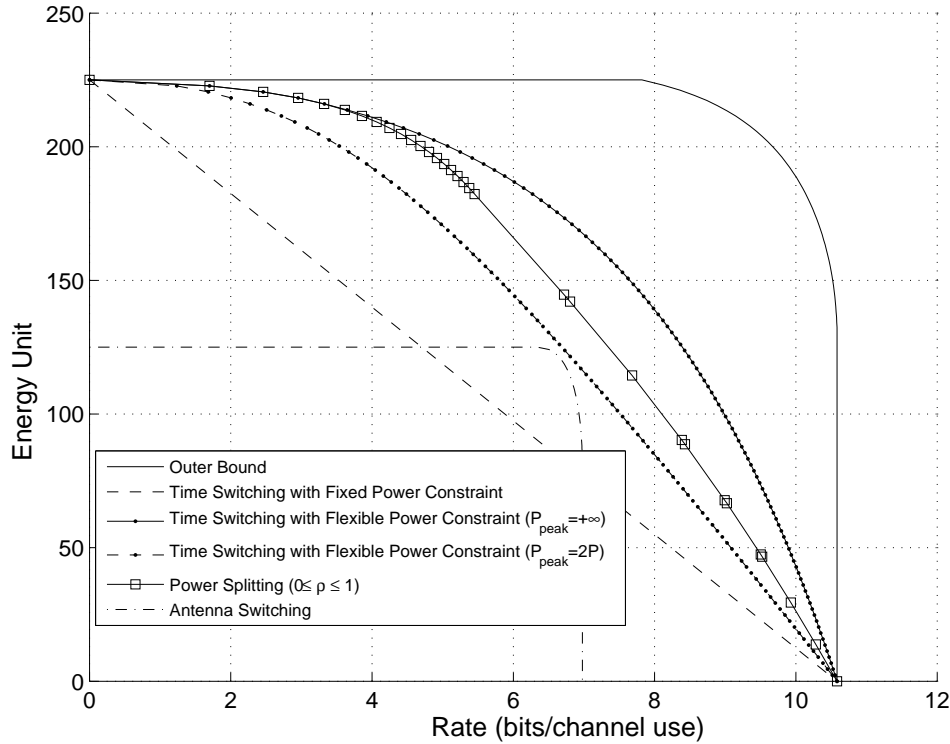}}}
\caption{Rate-energy tradeoff for a $2\times 2$ MIMO broadcast
system with co-located EH and ID receivers, and $\mv{H}=[1~ 0.5;
0.5~ 1]$.}\label{fig:RE region co-located receiver 1}
\end{figure}

\begin{figure}
\centering{
 \epsfxsize=5in
    \leavevmode{\epsfbox{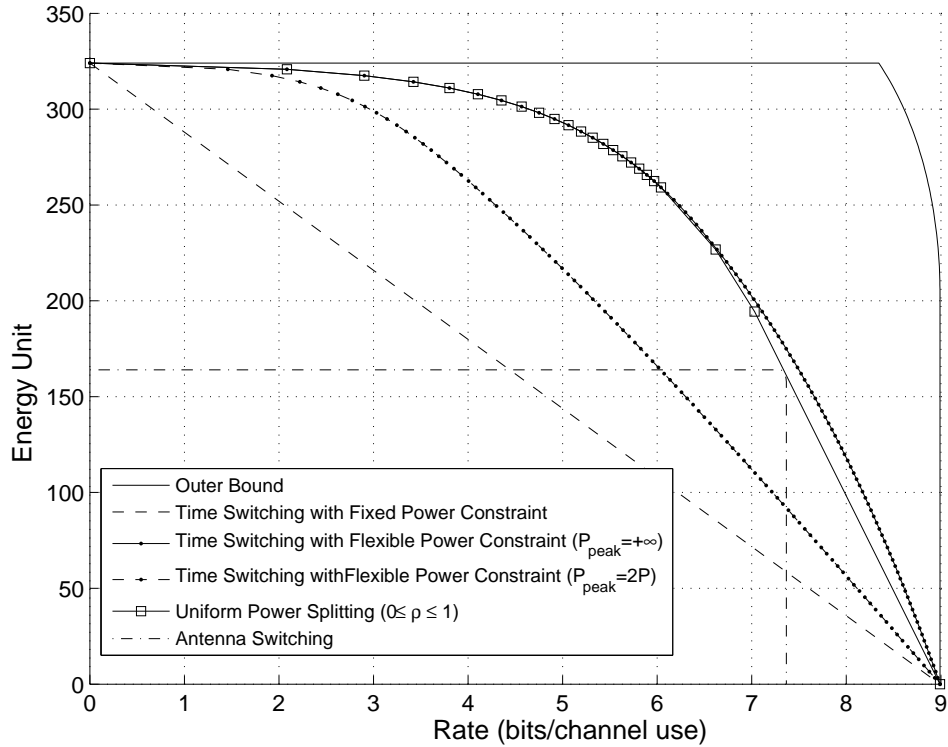}}}
\caption{Rate-energy tradeoff for a $2\times 2$ MIMO broadcast
system with co-located EH and ID receivers, and $\mv{H}=[1~ 0.8;
0.8~ 1]$.}\label{fig:RE region co-located receiver 2}
\end{figure}

\section{Concluding Remarks} \label{sec:conclusion}

This paper investigated the performance limits of emerging
``wireless-powered'' communication networks by means of
opportunistic energy harvesting from ambient radio signals or
dedicated wireless power transfer. Under a simplified three-node
setup, our study revealed some fundamental tradeoffs in designing
wireless MIMO systems for maximizing the efficiency of simultaneous
information and energy transmission. Due to the space limitation,
there are serval important issues unaddressed in this paper and left
for our future work, some of which are highlighted as follows:

\begin{itemize}
\item It will be interesting to extend the rate-energy region characterization
to more general MIMO broadcast systems with more than two receivers.
Depending on whether the energy and information receivers are
separated or co-located, and the broadcast information is for
unicasting or multicasting, various new problems can be formulated
for which the optimal solutions are challenging to obtain.

\item For the case of co-located energy and information
receivers, this paper shows a performance bound that in general
cannot be achieved by practical receivers. Although this paper has
shed some light on practical hardware designs to approach this limit
(e.g., by the power splitting scheme), further research endeavor is
still required to further reduce or close this gap, even for the
SISO AWGN channel.

\item In this paper, to simplify the analysis, it is assumed that the energy conversion efficiency at the energy receiver is
independent of the instantaneous amplitude of the received radio
signal, which is in general not true for practical RF energy
harvesting circuits \cite{Le}. Thus, how to design the broadcast
signal waveform, namely {\it energy modulation}, to maximize the
efficiency of energy transfer to multiple receivers under practical
energy conversion constraints is an open problem of high practical
interests.

\item Unlike the traditional view that
the receiver noise and/or co-channel interference degrade the
communication link reliability \cite{Jafar}, they are however
beneficial from the viewpoint of RF energy harvesting. Thus, there
exist nontrivial tradeoffs in allocating communication resources to
optimize the network interference levels for achieving maximal
information vs. energy transfer. More studies to reveal such
tradeoffs are worth pursuing.
\end{itemize}

\appendices

\section{Proof of Proposition \ref{proposition:opt S EH}} \label{appendix:proof opt S EH}

Without loss of generality, we can write the optimal solution to
Problem (P1) in its eigenvalue decomposition form as $\mv{S}_{\rm
EH}=\mv{V}\mv{\Sigma}\mv{V}^H$, where $\mv{V}\in\mathbb{C}^{M\times
M}$, $\mv{V}\mv{V}^H=\mv{V}^H\mv{V}=\mv{I}$, and
$\mv{\Sigma}=\mathtt{diag}(p_1,\ldots,p_M)$ with $p_1\geq p_2\geq
\ldots \geq p_M\geq 0$ and $\sum_{i=1}^Mp_i\leq P$. Let
$\hat{\mv{G}}=\mv{G}\mv{V}=[\hat{\mv{g}}_1, \ldots,
\hat{\mv{g}}_M]$. Then, the objective function of Problem (P1) can
be written as
$Q=\mathtt{tr}(\mv{G}\mv{S}\mv{G}^H)=\mathtt{tr}(\hat{\mv{G}}\mv{\Sigma}\hat{\mv{G}}^H)=\sum_{i=1}^Mp_i\|\hat{\mv{g}}_i\|^2\leq
P\|\hat{\mv{g}}_1\|^2$, where the equality holds if
$\|\hat{\mv{g}}_1\|^2=\max_i\|\hat{\mv{g}}_i\|^2$ and $p_1=P, p_i=0,
i=2,\ldots,M$. Let $\mv{V}=[\mv{v}_1,\ldots,\mv{v}_M]$. Since the
(reduced) SVD of $\mv{G}$ is given by
$\mv{G}=\mv{U}_G\mv{\Gamma}_G^{1/2}\mv{V}_G^H$ in Section
\ref{sec:system model}, we infer that $\|\hat{\mv{g}}_1\|^2$ is the
maximum of all $\|\hat{\mv{g}}_i\|^2$'s if and only if $\mv{v}_1$ is
the first column of $\mv{V}_G$ corresponding to the largest singular
value of $\mv{G}$, which is $\sqrt{g_1}$. Hence, we obtain the
optimal solution of Problem (P1) as $\mv{S}_{\rm
EH}=P\mv{v}_1\mv{v}_1^H$. The proof of Proposition
\ref{proposition:opt S EH} is thus completed.

\section{Proof of Theorem \ref{theorem:optimal S}} \label{appendix:proof optimal S}

The Lagrangian of (P3) can be written as
\begin{align}\label{eq:Lagrangian}
L(\mv{S},\lambda,\mu)=
\log\left|\mv{I}+\mv{H}\mv{S}\mv{H}^H\right|+\lambda\left(\mathtt{tr}\left(\mv{G}\mv{S}\mv{G}^H\right)-
 \bar{Q}\right) -\mu\left(\mathtt{tr}(\mv{S})-P \right).
\end{align}
Then, the Lagrange dual function of (P3) is defined as
$g(\lambda,\mu)=\max_{\smv{S}\succeq 0} L(\mv{S},\lambda,\mu)$, and
the dual problem of (P3), denoted as (P3-D), is defined as
$\min_{\lambda\geq 0, \mu\geq 0} ~g(\lambda,\mu)$. Since (P3) can be
solved equivalently by solving (P3-D), in the following, we first
maximize the Lagrangian to obtain the dual function with fixed
$\lambda\geq 0$ and $\mu\geq 0$, and then find the optimal dual
solutions $\lambda^*$ and $\mu^*$ to minimize the dual function. The
transmit covariance $\mv{S}^*$ that maximizes the Lagrangian to
obtain $g(\lambda^*,\mu^*)$ is thus the optimal primal solution of
(P3).

Consider first the problem of maximizing the Lagrangian over
$\mv{S}$ with fixed $\lambda$ and $\mu$. By discarding the constant
terms associated with $\lambda$ and $\mu$ in (\ref{eq:Lagrangian}),
this problem can be equivalently rewritten as
\begin{align}\label{eq:subproblem P3}
~\max_{\smv{S}\succeq 0} &
\log\left|\mv{I}+\mv{H}\mv{S}\mv{H}^H\right|-\mathtt{tr}\left(\left(\mu\mv{I}-\lambda\mv{G}^H\mv{G}\right)\mv{S}\right).
\end{align}

Recall that $g_1$ is the largest eigenvalue of the matrix
$\mv{G}^H\mv{G}$. We then have the following lemma.

\begin{lemma}\label{lemma:dual variable relationship}
For the problem in (\ref{eq:subproblem P3}) to have a bounded
optimal value, $\mu>\lambda g_1$ must hold.
\end{lemma}
\begin{proof}
We prove this lemma by contradiction. Suppose that  $\mu \leq
\lambda g_1$. Then, let $\mv{S}^{\star}=\beta\mv{v}_1\mv{v}_1^H$
with $\beta$ being any positive constant. Substituting
$\mv{S}^{\star}$ into (\ref{eq:subproblem P3}) yields
$\log(1+\beta\|\mv{H}\mv{v}_1\|^2)+\beta(\lambda g_1-\mu)$. Since
$\mv{H}$ and $\mv{G}$ are either independent (in the case of
separated receivers) or identical (in the case of co-located
receivers), it is valid to assume that $\|\mv{H}\mv{v}_1\|^2>0$ and
thus the value of the above function or the optimal value of Problem
(\ref{eq:subproblem P3}) becomes unbounded when $\beta\rightarrow
\infty$. Thus, the presumption that $\mu \leq \lambda g_1$ cannot be
true, which completes the proof.
\end{proof}

Since Problem (P3) should have a bounded optimal value, it follows
from the above lemma that the optimal primal and dual solutions of
(P3) are obtained when $\mu>\lambda g_1$. Let
$\mv{A}=\mu\mv{I}-\lambda\mv{G}^H\mv{G}$. It then follows that
$\mv{A}\succ 0$ with $\mu>\lambda g_1$, and thus $\mv{A}^{-1}$
exists. The problem in (\ref{eq:subproblem P3}) is then rewritten as
\begin{align}\label{eq:subproblem P3 new}
~\max_{\smv{S}\succeq 0} &
\log\left|\mv{I}+\mv{H}\mv{S}\mv{H}^H\right|-\mathtt{tr}\left(\mv{AS}\right).
\end{align}
Let the (reduced) SVD of the matrix $\mv{H}\mv{A}^{-1/2}$ be given
by
$\mv{H}\mv{A}^{-1/2}=\tilde{\mv{U}}\tilde{\mv{\Gamma}}^{1/2}\tilde{\mv{V}}^H$,
where $\tilde{\mv{U}}\in\mathbb{C}^{M\times T_2}$,
$\tilde{\mv{V}}\in\mathbb{C}^{M\times T_2}$,
$\tilde{\mv{\Gamma}}=\mathtt{diag}(\tilde{h}_1,\ldots,\tilde{h}_{T_2})$,
with $\tilde{h}_1\geq \tilde{h}_2\geq \ldots \geq
\tilde{h}_{T_2}\geq 0$. It has been shown in \cite{ZhangCR} under
the CR setup that the optimal solution to Problem
(\ref{eq:subproblem P3 new}) with arbitrary $\mv{A}\succ 0$ has the
following form:
\begin{align}\label{eq:optimal S appendix}
\mv{S}^{\star}=\mv{A}^{-1/2}\tilde{\mv{V}}\tilde{\mv{\Lambda}}\tilde{\mv{V}}^H\mv{A}^{-1/2}
\end{align}
where
$\tilde{\mv{\Lambda}}=\mathtt{diag}(\tilde{p}_1,\ldots,\tilde{p}_{T_2})$,
with $\tilde{p}_i=(1-1/\tilde{h}_i)^+, i=1,\ldots,T_2$.

Next, we address how to solve the dual problem (P3-D) by minimizing
the dual function $g(\lambda,\mu)$ subject to $\lambda\geq 0$,
$\mu\geq 0$, and the new constraint $\mu>\lambda g_1$. This can be
done by applying the subgradient-based method, e.g., the ellipsoid
method \cite{Boydnotes}, for which it can be shown (the proof is
omitted for brevity) that the subgradient of $g(\lambda,\mu)$ at
point $[\lambda,\mu]$ is given by $[
\mathtt{tr}(\mv{G}\mv{S}^{\star}\mv{G}^H)-\bar{Q},P-\mathtt{tr}(\mv{S}^{\star})]$,
where $\mv{S}^{\star}$ is given in (\ref{eq:optimal S appendix}),
which is the optimal solution of Problem (\ref{eq:subproblem P3})
for a given pair of $\lambda$ and $\mu$. When the optimal dual
solutions $\lambda^*$ and $\mu^*$ are obtained by the ellipsoid
method, the corresponding optimal solution $\mv{S}^{\star}$ for
Problem (\ref{eq:subproblem P3}) converges to the primal optimal
solution to Problem (P3), denoted by $\mv{S}^*$. The above
procedures for solving (P3) are summarized in Table \ref{table}. The
proof of Theorem \ref{theorem:optimal S} is thus completed.

\begin{table}
\centering \caption{Algorithm for Solving Problem (P3).}
\label{table}
\begin{tabular}{|l|}
\hline \hspace*{0.0cm} Initialize $\lambda\geq 0, \mu\geq 0,
\mu>\lambda g_1$
\\ \hspace*{0.0cm} Repeat \\
\hspace*{0.25cm} Compute $\mv{S}^{\star}$ using
(\ref{eq:optimal S appendix}) with the given $\lambda$ and $\mu$ \\
\hspace*{0.25cm}  Compute the subgradient of $g(\lambda,\mu)$  \\
\hspace*{0.25cm} Update $\lambda$ and $\mu$ using the ellipsoid
method subject to $\mu>\lambda g_1\geq 0$ \\
\hspace*{0.0cm} Until $\lambda$ and $\mu$ converge to the prescribed accuracy \\
\hspace*{0.0cm} Set $\mv{S}^*=\mv{S}^{\star}$ \\
\hline
\end{tabular}
\end{table}

\section{Proof of Corollary \ref{corollary:optimal S MISO}} \label{appendix:proof MISO optimal S}

Since $\mv{H}\equiv\mv{h}^H$, in Theorem \ref{theorem:optimal S},
the (reduced) SVD of $\mv{h}^H\mv{A}^{-1/2}$ with
$\mv{A}=\mu^*\mv{I}-\lambda^*\mv{G}^H\mv{G}$ simplifies to
$\mv{h}^H\mv{A}^{-1/2}=1\times\sqrt{\tilde{h}_1}\times\tilde{\mv{v}}_1^H$,
where $\tilde{h}_1=\|\mv{A}^{-1/2}\mv{h}\|^2$ and
$\tilde{\mv{v}}_1=\mv{A}^{-1/2}\mv{h}/\|\mv{A}^{-1/2}\mv{h}\|$.
Thus, from (\ref{eq:optimal S}) we have

\begin{align}
\mv{S}^*&=\mv{A}^{-1/2}\tilde{\mv{v}}_1\tilde{p}_1\tilde{\mv{v}}_1^H\mv{A}^{-1/2}
\\&=\frac{\mv{A}^{-1/2}\mv{A}^{-1/2}\mv{h}}{\|\mv{A}^{-1/2}\mv{h}\|}\left(1-\frac{1}{\|\mv{A}^{-1/2}\mv{h}\|^2}\right)^+\frac{\mv{h}^H\mv{A}^{-1/2}\mv{A}^{-1/2}}{\|\mv{A}^{-1/2}\mv{h}\|}
\\&=\mv{A}^{-1}\mv{h}\left(\frac{1}{\|\mv{A}^{-1/2}\mv{h}\|^2}-\frac{1}{\|\mv{A}^{-1/2}\mv{h}\|^4}\right)^+\mv{h}^H\mv{A}^{-1}.
\label{eq:optimal S MISO temp}
\end{align}
Moreover, since $T_2=1$ in this case, the maximum achievable rate is
given by
\begin{align}
R^*=\sum_{i=1}^{T_2}\log(1+\tilde{h}_i\tilde{p}_i)=\sum_{i=1}^{T_2}\left(\log(\tilde{h}_i)\right)^+=
\left(2\log\left(\|\mv{A}^{-1/2}\mv{h}\|\right)\right)^+.
\label{eq:optimal S MISO rate}
\end{align}
From (\ref{eq:optimal S MISO temp}) and (\ref{eq:optimal S MISO
rate}), Corollary \ref{corollary:optimal S MISO} thus follows.

\section{Proof of Corollary \ref{corollary:optimal S new}} \label{appendix:proof optimal S co-located
receivers}

Since $\mv{G}=\mv{H}$, from Theorem \ref{theorem:optimal S}, we have
$\mv{A}=\mu^*\mv{I}-\lambda^*\mv{G}^H\mv{G}=\mu^*\mv{I}-\lambda^*\mv{H}^H\mv{H}\succ
0$ (i.e., $\mu^*>\lambda^*h_1$). Recall that the (reduced) SVD of
$\mv{H}$ is given by $\mv{H}=\mv{U}_H\mv{\Gamma}_H^{1/2}\mv{V}_H^H$,
with $\mv{\Gamma}_H=\mathtt{ diag}(h_1,\ldots,h_{T_2})$, $h_1\geq
h_2\geq \ldots \geq h_{T_2}\geq 0$. Thus, it follows that
$\mv{A}=\mu^*\mv{I}-\lambda^*\mv{H}^H\mv{H}=\mv{V}_H(\mu^*\mv{I}-\lambda^*\mv{\Gamma}_H)\mv{V}_H^H$,
and
$\mv{A}^{-1/2}=\mv{V}_H(\mu^*\mv{I}-\lambda^*\mv{\Gamma}_H)^{-1/2}\mv{V}_H^H$.
Then, the (reduced) SVD of the matrix $\mv{H}\mv{A}^{-1/2}$ is given
by
$\mv{H}(\mv{V}_H(\mu^*\mv{I}-\lambda^*\mv{\Gamma}_H)\mv{V}_H^H)^{-1/2}=\mv{U}_H\mv{\Gamma}_H^{1/2}(\mu^*\mv{I}-\lambda^*\mv{\Gamma}_H)^{-1/2}\mv{V}_H^H$.
Since in Theorem \ref{theorem:optimal S}, the SVD of
$\mv{H}\mv{A}^{-1/2}$ is denoted by
$\tilde{\mv{U}}\tilde{\mv{\Gamma}}^{1/2}\tilde{\mv{V}}^H$, we thus
obtain $\tilde{\mv{U}}=\mv{U}_H$,
$\tilde{\mv{\Gamma}}=\mv{\Gamma}_H(\mu^*\mv{I}-\lambda^*\mv{\Gamma}_H)^{-1}$,
and $\tilde{\mv{V}}=\mv{V}_H$. From (\ref{eq:optimal S}), it then
follows that
\begin{align}
\mv{S}^*&=\mv{A}^{-1/2}\tilde{\mv{V}}\tilde{\mv{\Lambda}}\tilde{\mv{V}}^H\mv{A}^{-1/2}
\\
&=\mv{V}_H(\mu^*\mv{I}-\lambda^*\mv{\Gamma}_H)^{-1/2}\mv{V}_H^H\mv{V}_H\tilde{\mv{\Lambda}}\mv{V}_H^H\mv{V}_H(\mu^*\mv{I}-\lambda^*\mv{\Gamma}_H)^{-1/2}\mv{V}_H
\\
&=\mv{V}_H(\mu^*\mv{I}-\lambda^*\mv{\Gamma}_H)^{-1}\tilde{\mv{\Lambda}}\mv{V}_H^H
\\ &\triangleq \mv{V}_H\mv{\Sigma}\mv{V}_H^H \label{eq:optimal S temp}
\end{align}
where
$\mv{\Sigma}=(\mu^*\mv{I}-\lambda^*\mv{\Gamma}_H)^{-1}\tilde{\mv{\Lambda}}\triangleq\mathtt{diag}(\hat{p}_1,\ldots,\hat{p}_{T_2})$.
Note that in Theorem \ref{theorem:optimal S},
$\tilde{\mv{\Lambda}}=\mathtt{diag}(\tilde{p}_1,\ldots,\tilde{p}_{T_2})$,
with $\tilde{p}_i=(1-1/\tilde{h}_i)^+, i=1,\ldots,T_2$, and
$\tilde{\mv{\Gamma}}=\mathtt{diag}(\tilde{h}_1,\ldots,\tilde{h}_{T_2})=\mv{\Gamma}_H(\mu^*\mv{I}-\lambda^*\mv{\Gamma}_H)^{-1}$.
Thus, we obtain that
\begin{align}
\hat{p}_i&= \frac{1}{\mu^*-\lambda^* h_i}\left(1-\frac{\mu^*-\lambda^* h_i}{h_i}\right)^+ \\
&=\left(\frac{1}{\mu^*-\lambda^* h_i}-\frac{1}{h_i}\right)^+, \ \
i=1,\ldots,T_2. \label{eq:optimal power temp}
\end{align}
Moreover, it is easy to verify that
$\mv{\Gamma}_H\mv{\Sigma}=\tilde{\mv{\Gamma}}\tilde{\mv{\Lambda}}$.
Since for Problem (P3), the maximum achievable rate is given by
$R^*=\sum_{i=1}^{T_2}\log(1+\tilde{h}_i\tilde{p}_i)$, it follows
that
\begin{align}\label{eq:optimal rate temp}
R^*=\sum_{i=1}^{T_2}\log(1+h_i\hat{p}_i).
\end{align}

With (\ref{eq:optimal S temp}), (\ref{eq:optimal power temp}), and
(\ref{eq:optimal rate temp}), the proof of Corollary
\ref{corollary:optimal S new} is thus completed.

\section{Proof of Proposition \ref{proposition:TS zero alpha}} \label{appendix:proof zero alpha}

Due to orthogonal transmissions for the EH and ID receivers in the
TS scheme, we first show that the minimum transmission energy
consumed to achieve any harvested power $Q<Q_{\max}$ in the EH time
slot is equal to $Q/h_1$ regardless of $\alpha$, as follows: From
Section \ref{sec:system model} (assuming $\mv{G}=\mv{H}$), it
follows that the optimal $\mv{S}_2$ is in the form of
$q\mv{v}_1\mv{v}_1^H$, where $q>0$ and $\mv{v}_1$ is the eigenvector
of the matrix $\mv{H}^H\mv{H}$ corresponding to its largest
eigenvalue denoted by $h_1$. To achieve $Q$, it follows from
(\ref{eq:RE region time switching 2}) that
$\alpha\mathtt{tr}(\mv{H}\mv{S}_2\mv{H}^H)=Q$ and thus
$q=Q/(h_1\alpha)$. Thus, the minimum energy consumed to achieve $Q$
in (\ref{eq:RE region time switching 2}) is given by
$\alpha\mathtt{tr}(\mv{S}_2)=\alpha\cdot q= Q/h_1$, independent of
$\alpha$.

With this result, in (\ref{eq:RE region time switching 2}), the
transmission rate $R$ is given by
$(1-\alpha)\log|\mv{I}+\mv{H}\mv{S}_1\mv{H}^H |$ subject to
$(1-\alpha)\mathtt{tr}(\mv{S}_1)\leq (P-Q/h_1)$. Due to the
concavity of the $\log(\cdot)$ function, it follows that $R$ is
maximized when $\alpha\rightarrow 0$, under which the optimal
solution of $\mv{S}_1$ can be obtained similarly as for Problem
(P2). Thus, by changing the values of $Q$ in the interval of $0< Q <
Q_{\max}$ and solving the above problem with $\alpha=0$, the
corresponding maximum achievable rates as well as  the boundary of
$\mathcal{C}_{\rm R-E}^{\rm TS_2}(P)$ are obtained as given in
(\ref{eq:RE region time switching 2 new}) for the case of flexible
power constraint. Proposition \ref{proposition:TS zero alpha} thus
follows.

\section{Proof of Proposition \ref{proposition:SIMO}} \label{appendix:SIMO}

Since $\mv{H}\equiv\mv{h}\triangleq[h_1,\ldots,h_{N}]^T$, for any
set of $\rho_i$'s with $0\leq \rho_i\leq 1$, the harvested power is
equal to $Q=P\sum_{i=1}^{N}\rho_i|h_i|^2$. Clearly, $0\leq Q\leq
\|\mv{h}||^2P$. The equivalent SIMO channel for decoding information
then becomes
$\tilde{\mv{h}}\triangleq[\sqrt{1-\rho_1}h_1,\ldots,\sqrt{1-\rho_{N}}h_{N}]^T$.
Since for the SIMO channel, the transmit covariance matrix degrades
to a scalar equal to $P$, the maximum achievable rate is given by
(via applying the MRC beamforming at ID receiver):
\begin{align}
R=&\log\left(1+\|\tilde{\mv{h}}\|^2P\right) \\
=&\log\left(1+\sum_{i=1}^{N}(1-\rho_i)|h_i|^2P\right)\\ =&\log\left(1+\sum_{i=1}^{N}|h_i|^2P-\sum_{i=1}^{N}\rho_i|h_i|^2P\right)\\
=&\log\left(1+\|\mv{h}\|^2P-Q\right).
\end{align}
We thus have $\mathcal{C}_{\rm R-E}^{\rm PS}(P)=\{(R,Q): R\leq
\log(1+(\|\mv{h}\|^2P-Q)), 0\leq Q \leq \|\mv{h}\|^2P\}$.
Furthermore, since the above proof is valid for any $\rho_i$'s and
changing $\rho$ from 0 to 1 yields the value of
$Q=P\rho\sum_{i=1}^{N}|h_i|^2$ from 0 to $\|\mv{h}||^2P$, it thus
follows that $\mathcal{C}_{\rm R-E}^{\rm UPS}(P)=\{(R,Q): R\leq
\log(1+(\|\mv{h}\|^2P-Q)), 0\leq Q \leq \|\mv{h}\|^2P\}$, which is
the same as  $\mathcal{C}_{\rm R-E}^{\rm PS}(P)$. The proof of
Proposition \ref{proposition:SIMO} is thus completed.

\section{Proof of Proposition \ref{proposition:TS UPS compare}} \label{appendix:TS UPS compare}

First, we prove the former part of Proposition \ref{proposition:TS
UPS compare}, i.e., for any $P>0$, $\mathcal{C}_{\rm R-E}^{\rm
TS_1}(P)\subseteq\mathcal{C}_{\rm R-E}^{\rm UPS}(P)\subseteq
\mathcal{C}_{\rm R-E}^{\rm TS_2}(P)$. The proof of $\mathcal{C}_{\rm
R-E}^{\rm TS_1}(P)\subseteq\mathcal{C}_{\rm R-E}^{\rm UPS}(P)$ is
trivial, since the boundary of $\mathcal{C}_{\rm R-E}^{\rm TS_1}(P)$
is simply a straight line connecting the two boundary points
$(0,Q_{\max})$ and $(R_{\max},0)$ (cf. Fig. \ref{fig:RE region
co-located receiver 1}), and $\mathcal{C}_{\rm R-E}^{\rm UPS}(P)$ is
a convex set containing these two points. Next, we prove
$\mathcal{C}_{\rm R-E}^{\rm UPS}(P)\subseteq \mathcal{C}_{\rm
R-E}^{\rm TS_2}(P), \forall P\geq 0 $, by showing that for any given
harvested power $0<Q<Q_{\max}$, the corresponding boundary rate for
$\mathcal{C}_{\rm R-E}^{\rm TS_2}(P)$, denoted by $R_{\rm TS}$, is
no smaller than that for $\mathcal{C}_{\rm R-E}^{\rm UPS}(P)$,
denoted by $R_{\rm UPS}$, i.e., $R_{\rm TS}\geq R_{\rm UPS}$, as
follows: For any given $Q$, from the proof of Proposition
\ref{proposition:TS zero alpha}, it follows that  $R_{\rm TS}$ is
obtained (with $\alpha=0$) by maximizing
$\log|\mv{I}+\mv{H}\mv{S}_1\mv{H}^H |$ subject to
$\mathtt{tr}(\mv{S}_1)\leq (P-Q/h_1)$. On the other hand, for the
UPS scheme, from the harvested power constraint
$\rho\mathtt{tr}\left(\mv{H}\mv{S}\mv{H}^H\right)\geq Q$, it follows
that $\rho\geq Q/(h_1P)$ must hold. Note that $R_{\rm UPS}$ is
obtained by maximizing
$\log\left|\mv{I}+(1-\rho)\mv{H}\mv{S}\mv{H}^H\right|$ subject to
$\mathtt{tr}(\mv{S})\leq P$. Let $\mv{S}'=(1-\rho)\mv{S}$. The above
problem then becomes equivalent to maximizing
$\log\left|\mv{I}+\mv{H}\mv{S}'\mv{H}^H\right|$ subject to
$\mathtt{tr}(\mv{S}')\leq (1-\rho)P$. Since $\rho\geq Q/(h_1P)$, it
follows that $\mathtt{tr}(\mv{S}')\leq (1-Q/(h_1P))P=P-Q/h_1$. Thus,
it follows that $R_{\rm TS}\geq R_{\rm UPS}$. The former part of
Proposition \ref{proposition:TS UPS compare} is proved.

Next, we show the latter part of Proposition \ref{proposition:TS UPS
compare}, i.e., $\mathcal{C}_{\rm R-E}^{\rm UPS}(P)=\mathcal{C}_{\rm
R-E}^{\rm TS_2}(P)$ iff $P\leq (1/h_2-1/h_1)$. Consider first the
proof of the ``if'' part. For any $0<Q<Q_{\max}$, since $R_{\rm
TS}=\max_{\smv{S}_1}\log|\mv{I}+\mv{H}\mv{S}_1\mv{H}^H |$ subject to
$\mathtt{tr}(\mv{S}_1)\leq (P-Q/h_1)<(1/h_2-1/h_1)$, the optimal
solution for this problem must be beamforming, i.e.,
$\mv{S}_1=(P-Q/h_1)\mv{v}_1\mv{v}_1^H$ with $\mv{v}_1$ being the
eigenvector of $\mv{H}^H\mv{H}$ corresponding to its largest
eigenvalue $h_1$, due to the WF power application given by
(\ref{eq:WF}). Thus, it follows that $R_{\rm TS}=\log(1+h_1P-Q)$.
Consider now the UPS scheme. Suppose that
$\mv{S}=P\mv{v}_1\mv{v}_1^H$ and $\rho=Q/(h_1P)$. It then follows
that for UPS, the harvested power is equal to
$\rho\mathtt{tr}\left(\mv{H}\mv{S}\mv{H}^H\right)=Q/(h_1P)\times
(h_1P)=Q$, and the achievable rate  $R_{\rm UPS}$ is equal to
$\log\left|\mv{I}+(1-\rho)\mv{H}\mv{S}\mv{H}^H\right|=\log(1+(1-Q/(h_1P))\times
h_1P)=\log(1+h_1P-Q)=R_{\rm TS}$. Thus, we prove that
$\mathcal{C}_{\rm R-E}^{\rm UPS}(P)\supseteq \mathcal{C}_{\rm
R-E}^{\rm TS_2}(P)$. Since from the proof of the former part of
Proposition \ref{proposition:TS UPS compare} we have that
$\mathcal{C}_{\rm R-E}^{\rm UPS}(P)\subseteq \mathcal{C}_{\rm
R-E}^{\rm TS_2}(P)$, it thus follows that $\mathcal{C}_{\rm
R-E}^{\rm UPS}(P)=\mathcal{C}_{\rm R-E}^{\rm TS_2}(P)$. The ``if''
part is proved.

Second, we prove the ``only if'' part by contradiction. Suppose that
$\mathcal{C}_{\rm R-E}^{\rm UPS}(P)=\mathcal{C}_{\rm R-E}^{\rm
TS_2}(P)$ for $P=1/h_2-1/h_1+\delta$ with arbitrary $\delta>0$ and
thus $P> (1/h_2-1/h_1)$. For any $0<Q<h_1\delta$, the corresponding
$R_{\rm TS}$ on the boundary of $\mathcal{C}_{\rm R-E}^{\rm
TS_2}(P)$ is given by
$\max_{\smv{S}_1}\log|\mv{I}+\mv{H}\mv{S}_1\mv{H}^H|$, where
$\mathtt{tr}(\mv{S}_1)=P-Q/h_1>P-\delta=1/h_2-1/h_1$. Since
$\mathtt{tr}(\mv{S}_1)>1/h_2-1/h_1$, from the WF power allocation in
(\ref{eq:WF}) it follows that the optimal rank of $\mv{S}_1$ must be
greater than one. Consider now the UPS scheme. Since
$\mathcal{C}_{\rm R-E}^{\rm UPS}(P)=\mathcal{C}_{\rm R-E}^{\rm
TS_2}(P)$, it follows that for the same $Q$ as in the TS scheme, the
maximum rate for the UPS scheme is $R_{\rm UPS}= R_{\rm TS}$. From
the proof of the former part of Proposition \ref{proposition:TS UPS
compare}, we know that to achieve $Q$ with UPS, $\rho\geq Q/(h_1P)$
with the equality only when the transmit covariance is rank-one.
Furthermore, since $R_{\rm UPS}$ is equal to
$\log\left|\mv{I}+\mv{H}\mv{S}'\mv{H}^H\right|$ with
$\mathtt{tr}(\mv{S}')=(1-\rho)P\leq P-Q/h_1$, where the equality
holds only if $\mv{S}'$ is rank-one. Since $R_{\rm UPS}= R_{\rm
TS}$, it thus follows that $\mv{S}'=\mv{S}_1$ and
$\mathtt{tr}(\mv{S}')= P-Q/h_1$ must hold at the same time. Since
these two equalities require that $\mv{S}'$ have the rank greater
than one and equal to one, respectively, they cannot hold at the
same time. Thus, $R_{\rm UPS}= R_{\rm TS}$ cannot be true and the
presumption that $\mathcal{C}_{\rm R-E}^{\rm
UPS}(P)=\mathcal{C}_{\rm R-E}^{\rm TS_2}(P)$ does not hold. The
``only if'' part is proved.

Combining the proofs for both the ``if'' and ``only if'' parts, the
latter part of Proposition \ref{proposition:TS UPS compare} is thus
proved.

\end{document}